\documentclass{article}

\usepackage{arxiv}
\usepackage{amssymb}
\usepackage{appendix}
\usepackage[inline]{enumitem}
\usepackage[utf8]{inputenc}
\usepackage{amsmath}

\usepackage{amsfonts}
\usepackage{natbib}
\usepackage{color}
\usepackage{geometry}
\usepackage{mathtools}
\usepackage[utf8]{inputenc} % allow utf-8 input
\usepackage[T1]{fontenc}    % use 8-bit T1 fonts
\usepackage{hyperref}       % hyperlinks
\usepackage{url}            % simple URL typesetting
\usepackage{booktabs}       % professional-quality tables
\usepackage{amsfonts}       % blackboard math symbols
\usepackage{nicefrac}       % compact symbols for 1/2, etc.
\usepackage{microtype}      % microtypography
\usepackage{lipsum}
\newtheorem{theorem}{Theorem}

\newtheorem{algorithm}[theorem]{Algorithm}
\newtheorem{assumption}{Assumption}

\newtheorem{ex}[theorem]{Example}

\newtheorem{proposition}{Proposition}
\newtheorem{remark}{Remark}

\newenvironment{proof}[1][Proof]{\noindent\textbf{#1.} }{\ \rule{0.5em}{0.5em}}
\usepackage{appendix}
\usepackage{booktabs}
\usepackage{graphicx}
\usepackage{booktabs} % For formal tables
\usepackage{accents}
\usepackage{mwe}
\usepackage{bbm}
\allowdisplaybreaks
 \usepackage{algpseudocode} \usepackage[ruled]{algorithm}
 \algdef{SE}[DOWHILE]{Do}{doWhile}{\algorithmicdo}[1]{\algorithmicwhile\ #1}%
\usepackage{blkarray}
\def\NK{{\lvert \mathcal{K} \rvert}}
\def\NL{{\lvert \mathcal{L} \rvert}}
\def\NE{{\lvert E \rvert}}
\def\NL{{\lvert \mathcal{L} \rvert}}
\newcommand{\NodeList}{\mathcal{N}}
\def\minimize{\mathop{\rm minimize}}
\def\maximize{\mathop{\rm maximize}}
\def\subto{{\rm subject \mbox{   }\rm to}}
\newcommand{\diag}{\mathop{\mathrm{diag}}}

\newcommand{\erickmodified}[1]{{\color{blue}#1}}
\newcommand{\erickcomment}[1]{\erickmodified{red}}
\newcommand{\utsavmodified}[1]{{\color{}#1}}
\newcommand{\utsavcomment}[1]{\utsavmodified{}}

\newcommand{\V}[1]{{{\boldsymbol #1}}}
\newcommand{\BF}[1]{{\V{#1}}}
\newcommand{\BFx}{\V{x}}

\newcommand{\BFz}{\V{z}}

\newcommand{\BFlambda}{\V{\lambda}}
\newcommand{\BFell}{\V{\ell}}
\newcommand{\BFs}{\V{s}}
\newcommand{\BFeta}{\V{\eta}}

\newcommand{\BFgamma}{\V{\gamma}}
\newcommand{\BFpsi}{\V{\psi}}
\newcommand{\BFsigma}{\V{\sigma}}
\newcommand{\BFvarphi}{\V{\varphi}}

\newcommand{\BFUpsilon}{\V{\Upsilon}}
\newcommand{\BFupsilon}{\V{\upsilon}}
\newcommand{\BFbeta}{\V{\beta}}
\newcommand{\BFd}{\V{d}}

\newcommand{\BFdelta}{\V{\delta}}
\newcommand{\BFh}{\V{h}}
\newcommand{\BFy}{\V{y}}
\newcommand{\BFc}{\V{c}}

\newcommand{\BFw}{\V{w}}
\newcommand{\BFI}{I}
\newcommand{\BFq}{\V{q}}

\newcommand{\BFDelta}{\V{\Delta}}
\newcommand{\BFu}{\V{u}}
\newcommand{\X}{\mathcal{X}}
\newcommand{\BFe}{\V{e}}
\newcommand{\BFzero}{\V{0}}
\newcommand{\BFone}{\V{1}}

\newcommand{\mymbox}[1]{\mbox{\scriptsize #1}}
\renewcommand{\Re}{\mathbb{R}}
\newcommand{\quoteIt}[1]{``#1''}

\newcommand{\Q}{\mathcal{Q}}
\newcommand{\LL}{\mathcal{L}}

\newcommand{\I}{\mathcal{I}}

\newcommand{\Z}{\mathcal{Z}}
\newcommand{\K}{\mathcal{K}}
\newcommand{\Indic}{\mathbbm{1}}
\def\Expect{{\mathbb E}}
\def\Prob{{\mathbb P}}
\def\NK{{\lvert \mathcal{K} \rvert}}
\def\NE{{\lvert E \rvert}}
\usepackage{tikz}
\usetikzlibrary{arrows,automata,positioning}
   \usetikzlibrary{arrows.meta,positioning}
\usepackage{pgfplots}
\pgfplotsset{width=7.5cm,compat=1.12}
\usepgfplotslibrary{fillbetween}
\pgfplotsset{every axis/.append style={
                     tick label style={font=\footnotesize}  
                    }}
\title{The value of randomized strategies in distributionally robust risk averse network interdiction games}

\author{
  Utsav Sadana \\
 Department of Decision Sciences\\ HEC Montr\'eal\\ Montr\'eal H3T 2A7, Canada\\
  \texttt{utsav.sadana@hec.ca} \\
   \And
 Erick Delage \\
 Department of Decision Sciences\\ HEC Montr\'eal\\ Montr\'eal H3T 2A7, Canada\\
  \texttt{erick.delage@hec.ca}
}

\begin{document}
\maketitle

\begin{abstract}
Conditional Value at Risk (CVaR) is widely used to account for the preferences of a risk-averse agent in the extreme loss scenarios.  To study the effectiveness of randomization in interdiction games with an interdictor that is both risk and ambiguity averse, we introduce a \textit{distributionally robust network interdiction game}  where the interdictor randomizes over the feasible interdiction plans in order to minimize the worst-case CVaR of the flow with respect to both the  unknown distribution of the capacity of the arcs and  his mixed strategy over interdicted arcs.  The flow player, on the contrary, maximizes the total flow in the network.  By using the budgeted uncertainty set, we control the degree of conservatism in the model and reformulate the interdictor's non-linear problem as a bi-convex optimization problem. For solving this problem to any given optimality level,  we devise a spatial branch and bound algorithm  that uses the McCormick inequalities and reduced reformulation linearization technique (RRLT) to obtain convex relaxation of the problem. We also develop a  column generation algorithm to identify the optimal support of the convex relaxation which is then used in the coordinate descent algorithm to determine the upper bounds. The efficiency and convergence of the spatial branch and bound algorithm is established in the numerical experiments. Further, our numerical experiments show that randomized strategies can have significantly better in-sample and out-of-sample performance than optimal deterministic ones. 
 
\end{abstract}
% keywords can be removed
 \section{Introduction} 
\label{sec:introduction}

 Game-theoretic models are increasingly being used to determine the best practice in security policy, e.g., ARMOR program in Los Angeles International Airport, see \cite{Pita2008},  \cite{Pita2009}, \cite{Kar2017}. Yet, in many real-world applications, the parameters of an agent's decision model, e.g., the  capacity of the arcs in a network interdiction game, can be undetermined. It was shown in \cite{Ben-Tal2000} that  perturbations in the parameters of a linear programming problem can render the solution to become infeasible or significantly suboptimal.  A popular technique to tackle this issue is Stochastic Programming (SP)  which assumes that the distribution of the parameters is known. 
However, in many instances, the distribution of the random parameters is not exactly known. This can lead to  post-decision disappointment referred to as the \textit{optimizer's curse} (see \cite{Smith2006}). Alternatively, classical robust optimization models do not use any distributional information on the random parameters. However, imposing that only the worst-case outcome is to be taken into account can often lead to over-conservative solutions.

 Distributionally Robust Optimization (DRO) acts as a compromise between SP and RO by requiring the probability distribution to lie in a distributional ambiguity set
 and seeking a solution that performs best according to the worst-case distribution. Hence, when the ambiguous distribution set is properly tuned, it can prevent both the post-decision disappointment of SP models and the over-conservatism of RO models.  Furthermore, the set can be calibrated in ways that will provide statistical guarantees on the out-of-sample performance of the DRO solution. Recently, it was shown in \cite{Delage2019} that it can be beneficial to employ randomization in non-convex DRO problems. For an ambiguity-averse risk-neutral decision maker, i.e. one that minimizes worst-case expected value, \cite{Delage2018} proposed an algorithm to identify such strategies in mixed-integer two-stage DRO problems. \cite{Bertsimas2016} also studied the value of randomization specifically in the context of a network interdiction game with known parameters and risk-neutral agents. Nevertheless, none of these works provide either theoretical or computational means of identifying optimal randomized solutions for agents that employ more general risk measures than expected value, such as the Conditional Value at Risk (CVaR) measure introduced in \cite{Rockafellar2000}. Such risk measures are especially relevant in security policy model's such as network interdiction problems where a decision maker, e.g., a law-enforcement agency, might be concerned by the possibility of incurring huge losses under certain scenarios, e.g. caused by a large flow of illegal drugs, weapons or money.
 %\tb{In such cases, a risk-averse decision maker can use a Conditional Value at Risk (CVaR) measure (\cite{Rockafellar2000}) which is a coherent risk measure in the sense of \cite{Artzner1999}, to quantify her preferences.} 

In this paper, we study a \textit{distributionally robust network interdiction game} where the interdictor employs a CVaR risk measure to model his risk aversion.
Our contributions can be summarized as follows: 
 \begin{itemize}
   \item  On the methodological side, we introduce for the first time ambiguity and risk aversion in  network interdiction games where the interdictor minimizes the worst-case CVaR over both the unknown distribution of the capacities of the arcs and the distribution of interdicted arcs. This is in sharp contrast with the work of \cite{Loizou2015} who considers an interdictor that employs CVaR to handle parameter uncertainty but an expected value to handle the uncertainty caused by his own randomized strategy. We show that the approach in \cite{Loizou2015} can in fact produce a solution that is stochastically strictly dominated by the solution of our proposed model.
     %We have also shown that the existing model of \cite{Loizou2015} for the class of  distributionally robust games with CVaR risk-averse players can identify strictly stochastically dominated policies as optimal.
\item On the algorithm side, we complement the work of  \cite{Delage2018} by designing the first algorithm that can identify optimal randomized strategies for an ambiguity averse risk averse agent, i.e. an agent that minimizes a worst-case convex risk measure other than the expected value. Our algorithm is based on a spatial branch and bound scheme (see \cite{Al-Khayyal1983}) embedded with the column generation (CG)  method. It will successfully identify high quality randomized solution for networks containing hundreds of nodes in a few minutes. We expect that these results should pave the way to developing algorithms for solving non-zero sum Stackelberg games where players use a CVaR risk measure, and are ambiguity-averse.

%minimizing agent distributionally robust our column generation (CG) algorithm, which uses reduced reformulation linearization technique (RRLT) and McCormick inequalities to compute the lower bounds, can efficiently identify near optimal solution in many random instances studied in this paper. Also, we have devised a spatial branch and bound algorithm (see \cite{Al-Khayyal1983}) embedded with the column generation  algorithm  that can solve our \textit{ distributionally robust  interdiction game}  to optimality.} Our approach should pave the way to studying the  class of non-zero sum Stackelberg games where players use a CVaR risk measure, and are ambiguity-averse.
\item On the empirical side, we provide evidence which indicates that a network interdictor can significantly benefit from randomization. We find that, while it is true that deterministic plans are often worst-case CVaR optimal, in instances where randomization strictly improves this objective, the improvement is significant, i.e. a 4\% improvement on average. Furthermore, for  those case, the out-of-sample performance of the randomized strategy also always strictly improves the out-of-sample performance and by as much as 19\% on average. %Here all the numbers are weighted average based on table 2 and 3.

%    \item   The CVaR risk-averse interdictor's problem extends the DRO model of \cite{Delage2018} where the decision maker is considered to be risk-neutral. We show using the numerical experiments that the interdictor can benefit from randomization. 
\end{itemize} 

The rest of the paper is organized as follows. Section \ref{sec:RL} gives an overview of the literature covering network interdiction games and DRO problems which are closely related to this work. We also briefly discuss the algorithms which have been proposed previously to solve non-convex optimization problems with bilinear constraints.  Section \ref{sec:RDRNI} defines our \textit{distributionally robust network interdiction game}. The robust counterpart of the DRO model is given in Section \ref{sec:sBB} where we also describe the  spatial branch and bound algorithm embedded with the CG algorithm that is used to solve the interdictor's problem to global optimality.  Numerical experiments are provided in  Section \ref{sec:simulate} to demonstrate the convergence and efficiency of our algorithm and illustrate that randomization can significantly improve both in-sample and out-of-sample performance. Concluding remarks are given in Section \ref{sec:Conclusions}.

\section*{Notations}
Vectors are expressed in bold and matrices are represented by capital letters. $\BFone$ and $\BFzero$ denote column vectors of $1$'s and $0$'s  respectively. The identity matrix is denoted by $\BFI$, and $\BFe_i$ captures its $i$-th column.  The set of all probability measures on a finite discrete  measurable space $(\X, F_{\X})$ is denoted by $\Delta \X \subseteq \Re_{+}^{\lvert \X \rvert}$, where $F_{\X}$ denotes all subsets of $\X.$
%%%%----
%%%%-----
%%%%----
%%%%-----
\section{Related Literature}\label{sec:RL}

It is well-known that the illegal flow of  drugs, weapons or  other hazardous substances poses a threat to the security of a nation, see \cite{Magliocca2019} and the references in it. The law-enforcement agencies aim to reduce their flow while the adversaries, which may be smugglers or terrorist organizations, try to increase it. Network flow problems  arise in diverse areas like transportation (\cite{Israeli2002}), military  manpower planning (\cite{Gass1991}), medicine (\cite{Assimakopoulos1987}). For an elaborate description of theory, algorithms and applications of network flow problems, refer to \cite{Ahuja1993}. One of their useful applications is in determining the optimal interdiction policies for a network operator. Numerous network flow models have been developed to study the network interdiction game where the defender interdicts a set of arcs on the network to minimize the flow while the adversary maximizes it, see \cite{Wood1993}, \cite{Cormican1998}, \cite{Smith2013}, \cite{Smith2019}).

In \cite{Jain2010}, it is shown that randomization can be useful in  security applications like patrolling of airports. Recently, \cite{Bertsimas2016} introduced a randomized network interdiction game  where only the interdictor can randomize. The authors assumed interdictor as well as the flow player to be risk-neutral, and the model parameters are known with certainty. In contrast to their model,  \cite{Lei2018} assumed  deterministic strategies for the players, and the effect of interdiction on the  capacity of each arc is random. Also, the players use CVaR risk measure. In  \cite{Atamturk2017}, the  capacities are assumed to be stochastic, and the interdictor minimizes the maximum flow-at-risk over a discrete set of actions.  

Clearly, the literature on interdiction has been limited to SP models. Also, the benefit of using  randomized strategies for a risk-averse agent in DRO problems has not been explored. \cite{Aghassi2006} have developed robust game theory models where players are expected utility maximizers, and have incomplete information on the true probability distribution of their payoffs.  In order to account for ambiguity aversion in non-cooperative games,  \cite{Loizou2015} has proposed a distributionally robust game theory model where players use worst-case expected CVaR measure to evaluate the performance of their strategies. 
We show, using an example, that a single player counterpart of his model  results in an optimal policy that is strictly  stochastically dominated by the strategy produced by our model. 

Alternatively, we propose the  \textit{distributionally robust network interdiction game} where the interdictor minimizes the worst-case CVaR of the flow with respect to both the unknown distribution of  the capacity of the arcs and his mixed strategy over the feasible interdiction plans. In our model, both the interdictor and flow player can randomize over their respective action sets. Similar to \cite{Janjarassuk2008}, we assume that the flow player can observe the capacity of each arc as well as the interdicted arcs before selecting her optimal strategy. However,  the success of interdiction is a Bernoulli random variable in \cite{Janjarassuk2008} while it is deterministic in our model.
 
 \cite{Hajinezhad2018} and \cite{Hajinezhad2019}) made separability assumptions  on the objective functions to solve non-convex non-smooth optimization problems with bilinear constraints for machine learning and signal processing applications.  We do not make those assumptions on the interdictor's bi-convex  optimization problem and we devise the spatial branch and bound algorithm (see \cite{Al-Khayyal1983}) that iteratively searches the feasible space of the problem to obtain global optimal solutions within a given tolerance. The rate of convergence of the algorithm relies on generating tighter lower and upper bounds.  In the literature, McCormick inequalities (see \cite{McCormick1976}) are commonly used to obtain convex relaxations of the non-convex non-linear programming problems with bilinear constraints.  In order to obtain a tighter bound, Reformulation Linearization Technique (RLT) was proposed in \cite{Sherali1992} wherein valid inequalities obtained by multiplying the pairs of feasible constraints are added to the relaxed problem. This increases the size of LP relaxation. \cite{Liberti2006} proposed the reduced RLT, which in contrast to RLT, gives an exact reformulation of original non-convex problem with an additional number of linear equality constraints. The authors also show that RRLT combined with McCormick inequalities can result in tighter relaxations than applying McCormick directly to the bilinear terms.

%%%%----
%%%%-----
\section{ Distributionally robust network interdiction game}\label{sec:RDRNI}
 Consider a directed graph $G=(V,E)$, where V and E denote the nodes and arcs, respectively.  Let $e=(i,j)$ represent an arc on G, $\delta^-(i)=\{(i,j)\rvert j \in V \}$ and $\delta^+(i)=\{(j,i)\rvert j \in V \}$ denote, respectively, the set of arcs leaving and entering node $i \in V$.   A flow in the graph G is denoted by a non-negative vector $\BFx \in \Re_+^{\NE}$ so that for each $e \in E$,  we have $x_e \leq c_e,$ where the  capacity of all arcs is denoted by $\BFc \in \Re_+^{\NE}$. The conservation of flow at each node is ensured by
 \begin{align}
    & \sum_{e \in \delta^-(i)}x_e-\sum_{e \in \delta^+(i)}x_e =0 && \forall i \not\in \{s,t\}, \label{Flow conservation}
 \end{align}
 where $s$ and $t$ denote the source and sink node, respectively. The flow player aims at maximizing the flow in the network.
The interdictor, on the other hand, aims at minimizing the worst-case CVaR of the flow with respect to both the unknown distribution of the  capacity of the arcs and his mixed strategy over the interdicted arcs.  We assume that the interdictor has a budget to remove $B$ arcs in the network. Let $\LL$ denote the finite set of feasible plans for the interdictor where $\LL=\{\BFell \in \{0,1\}^{|E|}\rvert \sum_{e} \ell_e \leq B\}$. Here, $\ell_e=1$ if the interdictor removes arc $e$ and $\ell_e=0$ if arc $e$ is not interdicted. 
 The distribution of the capacities of all the  arcs is only known to lie in set  $\Q$. We assume that the distribution is discrete with a set of scenarios $\K$ supported on $\{\BFc^{k}\}_{k \in \K}$. 
 
 Similar to the model in \cite{Bertsimas2016},
 we assume that the randomized strategy  of the interdictor is a probability distribution $\BFu$ over the  set $\LL$ where $\BFu \in \Delta \LL$.    For any interdiction plan
$\BFell$ and scenario $k$, the flow player solves the following problem
\begin{subequations}\label{eq:Primal_flowP}
 \begin{align}
     f_{\BFell, k} :=    \maximize_{\BFx} \quad &\BFd^T \BFx \\
      \subto  \quad & N \BFx=0  \label{eq:flow_balance_Incidence}\\
     \quad & 0 \leq \BFx \leq C^k (\BFone-\BFell), \label{eq:PrimFlow_Capacity}
 \end{align}
 \end{subequations}
where $\BFd^T\BFx=\sum_{e \in \delta^+(t)}x_e$,  $C^k=\diag(\BFc^k)$, and \eqref{eq:flow_balance_Incidence} is a shorthand notation for constraint \eqref{Flow conservation}.

 The interdictor solves the following distributionally robust network interdiction (DRNI) problem:
\begin{align}
\mbox{(DRNI)\;\;}\minimize_{\BFu\in \Delta \LL} \max_{\BFq\in \Q} \mbox{CVaR}^\alpha_{ \BFell\sim \BFu,k\sim \BFq}[ f_{\BFell, k}], \label{Network_equilibrium}
\end{align}
where CVaR is defined over the joint distribution of capacities and  interdicted arcs. Namely, with risk aversion parameter $\alpha=[0,1)$, CVaR is defined as
\begin{align*}
  \mbox{CVaR}^\alpha_{ \BFell\sim \BFu,k\sim \BFq}[ f_{\BFell, k}]:=  \inf_{\zeta} \; \zeta + \frac{1}{1-\alpha}\sum\limits_{\BFell}\sum\limits_{k}q_k u_\BFell[f_{\BFell, k}- \zeta]^+,
\end{align*}
where $[f_{\BFell, k}- \zeta]^+:=\max (f_{\BFell, k}- \zeta,0).$

\begin{remark}
 Since the set $\X:=\{\BFx\, \rvert\,  N\BFx=0,\; 0\leq\BFx\leq C^k(\BFone-\BFell)\}$ of all possible s-t flows is convex, randomization is not beneficial for the flow player if he considers minimizing a convex risk measure of $\BFd^T\BFx$, such as,  $\min_{\BFu_x \in \Delta \X}\mbox{CVaR}_{\BFx\sim \BFu_x}^{\,\alpha}[-\BFd^T\BFx]$ (see \cite{Delage2019}).
\end{remark}
 
Under this setting, one can show that the interdictor's DRNI problem can be cast as a bi-convex DRO problem.

\begin{proposition}\label{prop:interdictorDRO}
When $\mathcal{Q}$ is a convex set, the interdictor's DRNI problem
\eqref{Network_equilibrium} is equivalent to the following bi-convex DRO problem
   \begin{subequations}\label{epi}
  \begin{align}
\mathop{\minimize_{\BFu,\zeta, \BFDelta,t,\BFeta}} \quad & t \\
 \text{subject to } \quad &
     \zeta + \frac{1}{1-\alpha}\sum\limits_{\BFell}\sum\limits_{k}{q}_{k} \Delta_{\BFell,  k}\leq t \quad && \forall \BFq \in \Q \label{eq:DRO_non_linear_constraint}\\
 \quad &   \Delta_{\BFell,  k} \geq u_\BFell f_{\BFell, k}-\eta_{\BFell} && \forall \BFell  \in \LL,\, k \in \K   \label{eq:DRO_Delta}\\
 \quad & \eta_{\BFell} =u_\BFell \zeta && \forall \BFell \in \LL \label{eq:DRO_eta}\\
 \quad &  \Delta_{\BFell,  k} \geq 0  &&\forall \BFell  \in \LL,\, k \in \K \label{eq:DRO_Delta_positive}\\
   \quad & \BFu \geq 0 \label{eq:DRO_u_positive} \\
  \quad & \BFone^T \BFu=1  && \label{eq:DRO_u_sum_one}\\
  \quad & 0 \leq \zeta \leq \bar{\zeta},  && \label{eq:DRO_zeta_bound}
  \end{align}
  \end{subequations}
  where $\bar{\zeta}:=\max_{k\in\K} f_{\BFzero, k}$.
\end{proposition}
\begin{proof}
See Appendix \ref{proof:prop:interdictorDRO}.
\end{proof}

The above problem is non-convex due to the bilinear terms $u_\BFell \zeta.$ A second difficulty resides in having to compute $f_{\BFell, k}$ for each scenario $k \in \K$ and interdiction plan $\BFell \in \LL$ in order to solve \eqref{epi}. This will be addressed algorithmically in Section \ref{sec:sBB}.

In the following example, we show that the distributionally robust model proposed in \cite{Loizou2015} identifies interdiction strategies which are strictly stochastically dominated by other feasible strategies. This indicates that Loizou's approach is not well motivated for this class of problems.
\begin{ex}
Consider an agent trying to  reduce the maximum flow from point $s$ to point $t$ which are located on two respective sides of a river. The agent has a budget to interdict two routes. There are three routes available, $e \in\{ T1,\;T2,\;B\}$, where routes $T1$ and $T2$ use two different tunnels to pass the river, while route $B$ uses a bridge to do so. In normal traffic conditions, the capacities are, $\tau$, $\tau$, and $\epsilon$, with $\frac{2\tau}{3}<\epsilon<\tau$, for routes $T1$, $T2$, and $B$ respectively. Unfortunately, all three routes are susceptible to congestion on the day of interest. In the case of $T1$ and $T2$, it is known that the city is planning to do some repairs, which would decrease the flow by $\delta$, using one repair team but no information is available regarding which tunnel it will be; the identity of the selected tunnel is denoted by $r\in\{ 1,\;2 \}$. On the other hand, there is also a weather forecast that predicts $50\%$ chance of snow fall which would create the same decrease in flow, $\delta$, on the bridge used by route $B$. We let $\delta$ satisfy $2(\tau-\epsilon)<\delta\leq\epsilon$. The flow player wants to maximize the flow from point $s$ to point $t$. 

Let the set of feasible interdiction plans be defined as follows:
\begin{align*}
    \LL = \{ \{T1,B\},\{T2,B\},\{T1,T2\},\{T1\},\{T2\},\{B\}, \{\varnothing\}\}.
\end{align*}
The possible scenarios $k=\{1,2,3,4\}$ for the capacity of the three routes $T1$, $T2$ and $B$ are given by:
    \begin{align*}
    \BFc^1=    \begin{bmatrix}
        \tau-\delta\\\tau\\\epsilon-\delta
        \end{bmatrix},\;
       \BFc^2=        \begin{bmatrix}
        \tau\\\tau-\delta\\\epsilon-\delta
        \end{bmatrix},\;
       \BFc^3=        \begin{bmatrix}
        \tau-\delta\\\tau\\\epsilon
        \end{bmatrix},\;
        \BFc^4=       \begin{bmatrix}
        \tau\\\tau-\delta\\\epsilon
        \end{bmatrix},
    \end{align*}
    with respective probabilities $q_1,q_2,q_3,q_4,$ such that $q_1+q_2=0.5$ and $q_3+q_4=0.5$.

Here are the numerical details regarding $f(\BFell,k)$ which denotes the total flow when interdiction plan $\BFell$ is chosen, and scenario $k$ is realized:
\begin{align*}
f(\BFell,k):=\left\{\begin{array}{cl}\tau&\mbox{ if  $\BFell= \{T1,B\}$ and $k\in\{1,3\}$},\\\tau-\delta&\mbox{ if $\BFell= \{T1,B\}$ and $k\in\{2,4\}$},\\\tau-\delta&\mbox{ if  $\BFell= \{T2,B\}$ and $k\in\{1,3\}$},\\\tau&\mbox{ if $\BFell= \{T2,B\}$ and $k\in\{2,4\}$},\\\epsilon-\delta&\mbox{ if $\BFell=\{T1,T2\}$ and $k\in\{1,2\}$},\\\epsilon&\mbox{ if $\BFell=\{T1,T2\}$ and $k \in \{3,4\}$},\\\tau+\epsilon-\delta&\mbox{ if $\BFell=\{T1\}$ and $k\in\{1,4\}$},\\\tau+\epsilon-2\delta&\mbox{ if $\BFell=\{T1\}$ and $k = 2$},\\\tau+\epsilon&\mbox{ if $\BFell=\{T1\}$ and $k = 3$},\\\tau+\epsilon-2\delta&\mbox{ if $\BFell=\{T2\}$ and $k=1$},\\\tau+\epsilon-\delta&\mbox{ if $\BFell=\{T2\}$ and $k \in\{2,3\}$},\\\tau+\epsilon&\mbox{ if $\BFell=\{T2\}$ and $k = 4$},\\2\tau-\delta&\mbox{ if $\BFell=\{B\}$ and $k\in\{1,2,3,4\}$},\\2(\tau-\delta)+\epsilon&\mbox{ if $\BFell=\{\varnothing\}$ and $k \in \{1,2\}$},\\2\tau-\delta+\epsilon&\mbox{ if $\BFell=\{\varnothing\}$ and $k \in \{3,4\}$}.\end{array}\right.\,
.\end{align*}
Consider two potential strategies: 
\begin{align*}
    &\BFu^L=(0.5,0.5,0,0,0,0,0),& &\BFu^{SD}=(0,0,1,0,0,0,0),
\end{align*}
% \erickcomment{used to be: 
% \begin{align*}
%     &\BFu^L=(0,0.5,0.5,0,0,0,0),& &\BFu^{SD}=(1,0,0,0,0,0,0),
% \end{align*}}
where $\BFu^L$ denotes the strategy to interdict routes $(T1,B)$ with 50\% probability and routes $(T2,B)$ with  $50\%$ probability; $\BFu^{SD}$ denotes the strategy to interdict routes $(T1,T2)$ with probability $1$.

Our robust risk-averse approach $g^{SD}(\BFu):= \sup_{\BFq\in\Q}\mbox{CVaR}^\alpha_{\BFell\sim \BFu, k \sim \BFq}[f(\BFell,k)]$, with $\alpha\geq 50\%$  leads to the following evaluation:
\begin{align*}
g^{SD}(\BFu^L)&=\tau &\mbox{v.s.}&&g^{SD}(\BFu^{SD})&=\epsilon.
\end{align*}
Since $\epsilon<\tau$, we get that $\BFu^{SD}$, interdicting both routes $T1$ and $T2$, is optimal.

Alternatively, the approach in \cite{Loizou2015} can be summarized as minimizing $g^L(\BFu):=\\ \sup_{\BFq\in\Q}\mbox{CVaR}^\alpha_{k\sim\BFq} [\Expect_{\BFell\sim\BFu}[f(\BFell,k)]]$
and leads to the following evaluation
\begin{align*}
g^{L}(\BFu^L)&=\tau-\delta/2 &\mbox{v.s.}&&g^L(\BFu^{SD})&=\epsilon.
\end{align*}
Since $\delta>2(\tau-\epsilon)$, the optimal decision in this case is to implement $\BFu^{L}$, i.e. interdicting $\{T1,B\}$ with $50\%$ probability and  $\{T2,B\}$ with $50\%$ probability, is optimal. Yet, it is clear from a purely statistical point of view that $\BFu^{SD}$ strictly stochastically dominates $\BFu^L$ for all $\BFq \in \Q$. Specifically, we have that
\begin{align*}
 \forall t\in\Re,\;  \Prob_{\BFell \sim \BFu^L, k\sim \BFq}(f(\BFell,k)\geq t)&=0.5 \Indic_{t \leq\tau}+0.5\Indic_{t\leq \tau-\delta}\\
 &\geq 0.5\Indic_{t\leq\epsilon}+0.5\Indic_{t \leq\epsilon-\delta}= \Prob_{\BFell\sim\BFu^{SD}, k\sim \BFq }(f(\BFell,k)\geq t) \,,
\end{align*}
where $\Indic_{x}$ denotes the indicator function of set $x$, and where inequality is strict when $\epsilon<t\leq\tau$ or
$\epsilon-\delta<t\leq\tau-\delta$.
\end{ex}

%\section{Solving the DR network interdiction game}\label{sec:sBB}
\section{Solving the DRNI  problem}\label{sec:sBB}
In this section, we propose a numerical scheme for solving the DRNI problem. It is well-known that the tractability of a robust optimization problem depends on the structure of the uncertainty set (see \cite{Ben-Tal2008} and  \cite{Ben-Tal2015}). In particular, for simplicity of exposition, we will focus on the case where the distribution ambiguity set contains perturbed versions of a reference distribution $\hat{\BFq}\in \Delta \mathcal{K}$.

\begin{assumption}
Let $\Q$ be defined as follows:
   \begin{align*}
  \Q:= \{\BFq \in \Re^{\NK} \ \rvert \exists \BFz \in \Z(\Gamma),\, \ \BFq\geq 0, \sum\limits_{k=1}^{\NK} q_{k}=1, \BFq= \hat{\BFq} + \diag(\bar{\BFq})\BFz\},
  \end{align*}
  where $\hat{\BFq} \in\Delta \mathcal{K}$ is a reference distribution, e.g., $\hat{\BFq}=\frac{1}{\NK}\BFone$, while $\bar{\BFq}$ models the magnitude of potential perturbations. Moreover, the set of perturbations $\Z$ refers to the following \quoteIt{budgeted uncertainty set}:
\begin{align*}
\Z(\Gamma) = \{\BFz\in \Re^{\NK}|-1 \leq \BFz \leq 1, \sum\limits_{k=1}^{\NK} |z_k| \leq \Gamma\},
  \end{align*}
 where $\Gamma$ denotes the maximum number of terms of $\hat{\BFq}$ that can be perturbed.
 \end{assumption}

The budgeted uncertainty set introduced in \cite{Bertsimas2004}  provides  the flexibility to study the trade-off between robustness and performance. In \cite{Bertsimas2003} and  \cite{Atamturk2007},  a budgeted uncertainty set  was used to study network flow problems with data uncertainty. For a given budget, the maximum number of parameters which can deviate from their nominal values is controlled with the budget $\Gamma$. A valuable property of the budgeted uncertainty set is that the robust counterpart of a linear constraint is representable in a linear program. This implies that the non-linear bi-convex DRO problem \eqref{epi} can be reformulated as a finite dimensional bi-convex optimization problem.

%\subsection{Formulating the robust counterpart of constraint \eqref{eq:DRO_non_linear_constraint}} \label{sec:RC}

\begin{proposition}
\label{theorem:robust_counterpart}
The robust counterpart of the interdictor's DRNI problem presented in \eqref{epi} is given by
 \begin{subequations}\label{eq:DRIP:RC}
  \begin{align}
 \minimize_{\scriptsize \begin{array}{c}\BFu, t,\BFw,\BFw^-,\\\BFeta,\chi, \BFbeta, \Delta, \zeta\end{array}} \quad & t\\  
\subto \quad &  \zeta+\sum_{k'\in\K}w_{k'}+\sum_{k'\in\K}w_{k'}^-+\Gamma\chi \notag\\\qquad &+\sum_{k'\in\K}
 \hat{q}_{k'}\beta_{k'}+\frac{1}{1-\alpha}\sum\limits_{\BFell \in \LL} \Delta_{\BFell,  k}-\beta_k \leq t\quad && \forall k\in\K \,,\label{eq:DRIP:RC:C1}\\
     \quad & \chi\geq  \bar{q}_k\beta_k-w_k&& \forall k \in \K \label{eq:chi_plus}\\
    \quad & \chi\geq -\bar{q}_k\beta_k-w^-_k&&\forall k 
    \in \K \label{eq:chi_minus}\\
 \quad & \BFw\geq 0,\BFw^-\geq 0,\chi \geq 0\label{eq:w:chi}\\
 \quad & \eqref{eq:DRO_Delta}-\eqref{eq:DRO_zeta_bound}. \notag%\label{eq:epigra_all}
  \end{align}
  \end{subequations}
\end{proposition} 
\begin{proof}
See Appendix \ref{append:proof_robust_counterpart}.
\end{proof}

The procedure that we propose for solving problem \eqref{eq:DRIP:RC} is motivated by the observation that $\zeta$ and $\BFu$ are complicating variables. Indeed, when either $\zeta$ or $\BFu$ is fixed, the problem reduces to a linear program. Since $\zeta$ is also constrained to lie in a bounded interval $\bar{\I}:=[0,\,\bar{\zeta}]$, a spatial branch and bound scheme on $\zeta$ seems appropriate (see \cite{Chandraker2008}).

Our implementation of the spatial branch and bound algorithm  will rely on the existence of two operators. Namely, after defining the problem that we are interested in solving as 
\begin{align*}
g(\bar{\I}):=\min_{\scriptsize \begin{array}{c}\BFu, t,\BFw,\BFw^-,\\\BFeta,\chi, \BFbeta, \Delta, \zeta\end{array}} \quad & t\\ 
\subto \quad& \eqref{eq:DRO_Delta}-\eqref{eq:DRO_u_sum_one}, \eqref{eq:DRIP:RC:C1} -\eqref{eq:w:chi}, \nonumber\\
&  \zeta \in \bar{\I}\,,
\end{align*}
the algorithm will assume the existence of the following two bounding operators $g_{\mymbox{lb}}(\I)$ and $g_{\mymbox{ub}}(\I)$ which satisfy:
\[g_{\mymbox{lb}}(\I) \leq g(\I)\leq g_{\mymbox{ub}}(\I)\,,\,\forall \I\subseteq \bar{\I}\,,\]
and such that for all sequence of intervals $\I_1,\I_2,\dots$ converging to some $\zeta$, the associated sequence of bounds $(g_{\mymbox{lb}}(\I_j),g_{\mymbox{ub}}(\I_j))$ converge to $g(\zeta)$. Finally, the operator $g_{\mymbox{ub}}(\I)$ will be such that one can always efficiently produce a $\zeta_{\mymbox{ub}}^*(\I)\in\I$ such that $g(\{\zeta_{\mymbox{ub}}^*\})=g_{\mymbox{ub}}(\I)$.

With this in hand, we can describe the algorithm. First in words, the spatial branch and bound algorithm starts at a root node capturing $\bar{\I}$ and branches on this node by subdividing it into a number of sub-intervals, considered sub-nodes of the branch and bound tree.
 \footnote{Implementation details: for any interval $\tilde{\I}:=[\zeta_{\mymbox{lb}},\; \zeta_{\mymbox{ub}}]$,  we create sub-intervals $[\zeta_{\mymbox{lb}},\; \zeta_{\mymbox{lb}}+p(\zeta_{\mymbox{ub}}-\zeta_{\mymbox{lb}})]$ and $[ \zeta_{\mymbox{lb}}+p(\zeta_{\mymbox{ub}}-\zeta_{\mymbox{lb}}),\; \zeta_{\mymbox{ub}}]$, where $p=0.2$ if $ (\hat{\zeta}^*-\zeta_{\mymbox{lb}})< (\zeta_{\mymbox{ub}}-\hat{\zeta}^*)$; $p=0.8$ otherwise
 %if $ (\hat{\zeta}^*-\zeta_{\mymbox{lb}})> (\zeta_{\mymbox{ub}}-\hat{\zeta}^*)$ and $p=0.5$, otherwise.
 }  Nodes are progressively selected and branched upon until either at node $j$, we have that  $g_{\mymbox{lb}}(\I_j)$ is close enough to $g_{\mymbox{ub}}(\I_j)$, or  $g_{\mymbox{lb}}(\I_j)$ is larger than $g(\{\hat{\zeta}^*\})$ where $\hat{\zeta}^*$ is the best solution found so far. When no more nodes need to be branched upon, the algorithm can terminate and conclude that $\{\hat{\zeta}^*\}$ is close enough to being globally optimal. Based on  $\hat{\zeta}^*$ it is then possible to get a nearly optimal solution $\hat{\BFu}^*$ from problem \eqref{eq:DRIP:RC} where $\zeta$ is fixed to $\hat{\zeta}^*$. Finally, the exact worst-case CVaR of $\hat{\BFu}^*$ can be obtained by solving  problem \eqref{eq:DRIP:RC} where $\BFu$ is fixed to $\hat{\BFu}^*$. For clarity, Algorithm \ref{alg:sbandbAlgo} presents the pseudocode for the procedure that was described.

\begin{algorithm}
\caption{Spatial branch and bound algorithm for solving problem \eqref{eq:DRIP:RC}}\label{alg:sbandbAlgo}
\begin{algorithmic}[1]
\Procedure{SpatialBranch$\&$Bound}{$\epsilon$,$n$}
\State $UB^*\gets\infty$, $\NodeList \gets \{\bar{\I}\}$
\While{$\NodeList \neq \emptyset$}
\State Sort $\NodeList$ according to $g_{\mymbox{lb}}(\cdot)$% for all $I\in\NodeList$
\State Take first $\I$ out of $\NodeList$
%\State $\NodeList \gets \NodeList/\I$
\If{$g_{\mymbox{ub}}(\I) < UB^*$}
\State $UB^* \gets g_{\mymbox{ub}}(\I)$
\State $\hat{\zeta}^*\gets \zeta_{\mymbox{ub}}^*(\I)$
\EndIf
\State Remove from $\NodeList$, all $\I$ such that $g_{\mymbox{lb}}(\I)\geq UB^*$ 
\If{$g_{\mymbox{ub}}(\I)>(1+\epsilon) g_{\mymbox{lb}}(\I))$}
\State Divide $\I$ into $n$ sub-intervals $\{ \I_1,\dots,\I_n\}$
\State $\NodeList \gets \NodeList \cup \{ \I_1,\dots,\I_n\}$
\EndIf
\EndWhile
\State Solve problem \eqref{eq:DRIP:RC} with constraint $\zeta=\hat{\zeta}^*$ to get $\hat{\BFu}^*$
\State Solve problem \eqref{eq:DRIP:RC} with constraint $\BFu=\hat{\BFu}^*$ to get $\hat{t}^*$
\State \textbf{return} $\hat{\BFu}^*,\hat{t}^*$%,\BFw^*,\BFw^-^*,\BFeta^*,\mu^*,\chi^*, \BFbeta, \BFDelta, \zeta$
\EndProcedure
\end{algorithmic}
\end{algorithm}

We are left with describing how the two operators can be efficiently implemented.

  \subsection{Using  RRLT with C\&CG for \texorpdfstring{$g_{\mymbox{lb}}(\I)$}{lg}}
  \label{subsec:Convex_relax}
  In this section, we describe an efficient procedure that can be used to establish a lower bound for the optimal value of problem \eqref{eq:DRIP:RC}. This procedure will need to overcome the two underlying difficulties of problem \eqref{eq:DRIP:RC}, namely that constraint \eqref{eq:DRO_eta} is bilinear in $\BFu$ and $\zeta$, and that the size of this problem is exponential with respect to $|\K|$ due to the set $\LL$. To tackle the first obstacle, we will employ a popular reduced reformulation linearization technique (see \cite{Liberti2006}) that  will relax the problem to a linear program. The second obstacle will be dealt with by employing a column generation scheme  (\cite{Desrosiers2005}) that only considers a subset $\hat{\LL}\subset \LL$ and
  progressively adds relevant support points to it until optimality conditions are satisfied.
  
Starting with the idea of relaxing the problem to a linear program, we follow similar steps as used in \cite{Liberti2006}. Namely, we start by introducing a set of redundant constraints in problem \eqref{eq:DRIP:RC}. This gives rise to the following equivalent optimization model
\begin{subequations} \label{M_RRLT}
  \begin{align}
\minimize_{\scriptsize \begin{array}{c}\BFu, t,\BFw,\BFw^-,\\\BFeta,\chi, \BFbeta, \Delta, \zeta\end{array}} \quad & t \\ 
\subto  \quad & \eqref{eq:DRO_Delta}-\eqref{eq:DRO_u_sum_one}, \eqref{eq:DRIP:RC:C1}-\eqref{eq:w:chi}, \notag\\
  \quad & \sum\limits_{\BFell \in \LL} \eta_{\BFell}= \zeta  \label{RRLT} \\
    \quad & \eta_{\BFell} \geq  u_\BFell \zeta_{\mymbox{lb}}&& \forall \BFell  \in \LL \label{eq:mccor1}\\
  \quad & \eta_{\BFell} \leq u_\BFell\zeta_{\mymbox{ub}}&& \forall \BFell  \in \LL \label{eq:mccor2} \\
 \quad &\eta_{\BFell}\geq \zeta +\zeta_{\mymbox{ub}}(u_\BFell-1) && \forall \BFell  \in \LL \label{eq:mccor3}\\
 \quad & \eta_{\BFell} \leq \zeta+\zeta_{\mymbox{lb}}(u_\BFell-1) && \forall \BFell  \in \LL, \label{eq:mccor4}\\
\quad &  \zeta_{\mymbox{lb}} \leq \zeta \leq \zeta_{\mymbox{ub}}, \label{eq:zetabound}
  \end{align}
  \end{subequations}
  where $\zeta_{\mymbox{lb}}$ and $\zeta_{\mymbox{ub}}$ are the respective boundaries of $\I$. 
In problem \eqref{M_RRLT}, constraint \eqref{RRLT} is redundant since $\sum_{\BFell\in\LL}u_\ell=1$ implies that $\sum_{\BFell\in\LL}u_\ell \zeta=\zeta$, and we have that $u_\ell\zeta=\eta_\ell$. On the other hand, constraints \eqref{eq:mccor1}-\eqref{eq:mccor4} are so-called McCormick inequalities (see \cite{McCormick1976}) which are known to be redundant given that $\zeta\in\I$ and $0\leq \BFu\leq 1$.

We obtain the linear relaxation of the above problem by removing constraint \eqref{eq:DRO_eta} from problem \eqref{M_RRLT}. For completeness, we present this linear programming relaxation in full details below:
\begin{subequations} \label{RPM}
  \begin{align}
\minimize_{\scriptsize \begin{array}{c}\BFu, t,\BFw,\BFw^-,\\\BFeta,\chi, \BFbeta, \Delta, \zeta\end{array}} \quad & t\\
 \subto \quad & \zeta+\sum_{k'\in\K}w_{k'}+\sum_{k'\in\K}w_{k'}^-+\Gamma\chi \notag \\\qquad&+\sum_{k'\in\K} 
 \hat{q}_{k'}\beta_{k'}+\frac{1}{1-\alpha}\sum\limits_{\BFell \in \LL} \Delta_{\BFell,  k}-\beta_k \leq t && \forall k\in\K \,,\label{eq:DRIP:RC:C1:n}\\
   \quad & \chi\geq  \bar{q}_k\beta_k-w_k&& \forall k \in \K \label{eq:chi_plus:n}\\
    \quad & \chi\geq -\bar{q}_k\beta_k-w^-_k&&\forall k 
    \in \K \label{eq:chi_minus:n}\\
      \quad & \BFone^T \BFu=1  && \label{eq:DRO_u_sum_one:relaxed}\\
       \quad & \sum\limits_{\BFell \in \LL} \eta_{\BFell}= \zeta  \label{RRLT:relaxed}\\
 \quad &\eta_{\BFell}\geq \zeta +\zeta_{\mymbox{ub}}(u_\BFell-1) && \forall \BFell  \in \LL \label{eq:mccor3:n}\\
 \quad & \eta_{\BFell} \leq \zeta+\zeta_{\mymbox{lb}}(u_\BFell-1) && \forall \BFell  \in \LL, \label{eq:mccor4:n}\\
         \quad & \zeta_{\mymbox{lb}} \leq \zeta \leq \zeta_{\mymbox{ub}}  && \label{eq:DRO_zeta_bound:n}\\
            \quad & \BFw\geq 0,\BFw^-\geq 0,\chi \geq 0\label{eq:w:chi:n}\\
     \quad &   \Delta_{\BFell,  k} \geq u_\BFell f_{\BFell, k}-\eta_{\BFell} && \forall \BFell  \in \LL,\, k \in \K   \label{eq:DRO_Delta:n}\\
 \quad &  \Delta_{\BFell,  k} \geq 0,  &&\forall \BFell  \in \LL\, k \in \K \label{eq:DRO_Delta_positive:n}\\
   \quad & \BFu \geq 0 \label{eq:DRO_u_positive:n} \\
\quad & \eta_{\BFell} \geq u_\BFell\zeta_{\mymbox{lb}}&& \forall \BFell  \in \LL \label{eq:mccor1:n}\\
  \quad & \eta_{\BFell} \leq u_\BFell\zeta_{\mymbox{ub}}&& \forall \BFell  \in \LL. \label{eq:mccor2:n} 
  \end{align}
  \end{subequations}
 
In what follows, it will be useful to compactly represent \eqref{RPM} as follows:
\begin{subequations}\label{Compact_RMP}
\begin{align}
   \mathop{\minimize_{\BFx, \{\BFy_{\BFell}\}_{\BFell \in \LL}}} \quad & \BFh^T\BFx\\
 \subto   \quad & 
    A \BFx+\sum_{\BFell \in \LL} B_{\BFell}\BFy_{\BFell} \leq \BFs  \label{eq:xandy} \\
  \quad & W_{\BFell}\BFy_{\BFell} \leq \BFzero && \forall \BFell \in \LL, \label{eq:only:y}
\end{align}
\end{subequations}
where \eqref{eq:xandy} summarizes constraints \eqref{eq:DRIP:RC:C1:n}, \eqref{eq:chi_plus:n},  \eqref{eq:chi_minus:n}, \eqref{eq:DRO_u_sum_one:relaxed}, \eqref{RRLT:relaxed},  \eqref{eq:mccor3:n}, \eqref{eq:mccor4:n}, \eqref{eq:DRO_zeta_bound:n},  
 \eqref{eq:w:chi:n} while \eqref{eq:only:y} summarizes constraints \eqref{eq:DRO_Delta:n},   \eqref{eq:DRO_Delta_positive:n}, \eqref{eq:DRO_u_positive:n}, \eqref{eq:mccor1:n}, \eqref{eq:mccor2:n}.
The decision variables are given by
\begin{align*}
    &\BFx =[\BFbeta^T~\BFw^T\; {\BFw^-}^T~\zeta~\chi~t]^T,\\
    &\BFy_{\BFell} = [\BFDelta^T_{\BFell}~u_\BFell~\eta_{\BFell}]^T, &&\forall \BFell \in \LL
\end{align*}
where $\BFx \in \Re^{3(\NK+1)}, \BFy_\BFell \in \Re^{\NK+2}$ and  $\BFDelta_{\BFell} \in \Re^{\NK}$ for each $\BFell \in \LL$,\;    $A \in \Re^{(5\NK+2\NL+7) \times (3(\NK+1))}$,\;  $B_{\BFell} \in \Re^{(5\NK+2\NL+7)\times (\NK+2)}$, $W_{\BFell} \in \Re^{(2\NK+3) \times (\NK+2)}$, $\BFs \in \Re^{5\NK+2\NL+7}$, $\BFh \in \Re^{ 3(\NK+1)}$. 
For the full description of  $A$, $B$, $W$, $s$, $h$, refer to Appendix \ref{sec:appen:matrix}.
  
 The idea behind column generation methods is to consider that at optimality $\BFy_\BFell\neq 0$ only for a small set of index $\ell\in\LL$. This is a legitimate assumption for our DRNI problem where we expect that there should be an optimal strategy that only randomizes among a relatively small (non-exponential) number of interdiction plans. This was observed for instance in the distributionally robust risk neutral facility location problem studied in \cite{Delage2018}.
 
Given a set $\hat{\LL}\subseteq \LL$, by linear programming duality, we have that the solution of 
\begin{subequations}\label{eq:compactWithZero}
\begin{align}
   \mathop{\minimize_{\BFx, \{\BFy_{\BFell}\}_{\BFell \in \LL}}} \quad & \BFh^T\BFx\\
 \subto   \quad & 
    A \BFx+\sum_{\BFell \in \LL} B_{\BFell}\BFy_{\BFell} \leq \BFs \label{eq:compactWithZero:C1}\\
  \quad & W_{\BFell}\BFy_{\BFell} \leq \BFzero &&  \forall \BFell \in \hat{\LL} \label{eq:compactWithZero:C2}\\
  & \BFy_\ell = 0 &&\forall \ell\in\LL/\hat{\LL},\label{eq:compactWithZero:C3} 
\end{align}
\end{subequations}
  is optimal with respect to problem \eqref{Compact_RMP} if and only if a solution of its dual problem
  \begin{subequations}\label{eq:compactDual}
  \begin{align}
    \maximize_{\BFpsi,\, \{\BFsigma_{\BFell}\}_{\BFell \in \hat{\LL}} } \quad & -\BFpsi^T \BFs\\
  \subto  \quad & \BFh+A^T\BFpsi =\BFzero\\
  \quad & B^T_{\BFell}\BFpsi +W_{\BFell}^T\BFsigma_{\BFell}  = \BFzero && \forall \BFell \in \hat{\LL}\\
  & \BFpsi \geq 0, \;\BFsigma_\BFell \geq 0 && \forall \BFell \in \hat{\LL}, 
  \end{align}
  \end{subequations}
where $\BFpsi \in \Re^{ 5\NK+2\NL+7}$ and $\BFsigma_{\BFell} \in \Re^{2\NK+3}$ are the dual variables associated to constraints \eqref{eq:compactWithZero:C1} and \eqref{eq:compactWithZero:C2} respectively, can be completed with some $\BFsigma_{\BFell} \in \Re^{2\NK+3}$ for all $\BFell\in\LL/\hat{\LL}$ in a way that makes it feasible in the dual of problem \eqref{Compact_RMP}, i.e. problem \eqref{eq:compactDual} where $\hat{\LL}$ is replaced with $\LL$.

In particular, this can be verified after solving problem \eqref{eq:compactWithZero} for some $\hat{\LL}\subseteq \LL$ by obtaining a set $(\hat{\BFpsi},\{ \hat{\BFsigma}_{\BFell} \}_{\ell\in\hat{\LL}})$ of optimal dual variables for constraints \eqref{eq:compactWithZero:C1} and \eqref{eq:compactWithZero:C2} and verifying if they satisfy the following condition:
\begin{eqnarray}\label{eq:violatedConstraintCond}
\inf_{\BFell \in \LL }\sup_{ \BFsigma_{\BFell} \geq 0} \inf_{\BFy_{\BFell}} {\BFpsi^*}^T B_{\BFell}\BFy_{\BFell} + \BFsigma^T_{\BFell}W_{\BFell}\BFy_{\BFell} \geq 0.
\end{eqnarray}
Furthermore, when condition \eqref{eq:violatedConstraintCond} is not satisfied, a violating $\BFell\in\LL$, which is necessarily not in $\hat\LL$, can be identified and added to $\hat{\LL}$ in order to improve the solution obtained by problem \eqref{eq:compactWithZero}.

Two important observations need to be made at this point. First, fortunately enough problem \eqref{eq:compactWithZero} can be shown to reduce to a linear program which size is linear in $|\hat{\LL}|$ given that only the decision variables $(\BFx,\{ \BFy_{\BFell} \}_{\BFell\in\hat{\LL}})$ need to be optimized, while the only constraints indexed by some $\BFell\in\LL/\hat{\LL}$ are constraint \eqref{eq:compactWithZero:C3} and a subset of constraint \eqref{eq:compactWithZero:C1} which capture constraints \eqref{eq:mccor3:n} and \eqref{eq:mccor4:n}. The latter two become redundant for any $\BFell\notin\hat{\LL}$ since in those cases the constraints reduce to:
\begin{align*}
&\zeta_{\mymbox{lb}}\leq \zeta \leq \zeta_{\mymbox{ub}} && \forall \BFell  \in \LL/\hat{\LL}.
\end{align*}
Secondly, one can also show that condition \eqref{eq:violatedConstraintCond} can be verified efficiently by solving a mixed-integer linear program of reasonable size as described in the following proposition.

\begin{proposition}\label{thm:violatedConstReform}
Given some $\I\subseteq \bar{\I}$, verifying condition \eqref{eq:violatedConstraintCond} is equivalent to verifying whether the optimal value of the following mixed-integer linear program is non-negative:
\begin{subequations}\label{eq:violatedConstMILP}
\begin{align}
    \mathop{\minimize_{\BFell \in \LL, \BFDelta, \eta}}_{\{\BFlambda_k, \BFupsilon_k,\BFUpsilon_k\}_{k\in\K}} \quad & \frac{{\BFvarphi^*}^T}{1-\alpha} \BFDelta+p^* +\pi^*
    \eta \\
    \subto \quad &  \Delta_{k}\geq \BFc_k^T(\BFlambda_k- \BFUpsilon_k) -\eta &&\forall k \in \K\\
  \quad &  \BFUpsilon_k \leq \BFell &&k \in \K\\
  \quad &  \BFUpsilon_k \leq \BFlambda_k &&\forall k \in \K\\
   \quad & \BFUpsilon_k \geq \BFlambda_k+\BFell -1&&\forall k \in \K\\
  \quad &  \BFUpsilon_k  \geq \BFzero &&\forall k \in \K\\
    \quad & \BFlambda_k+N^T\BFupsilon_k -\BFd \geq 0 && \forall k \in \K \label{sub_lambda}\\
   \quad & 0\leq \BFlambda_k\leq 1 && \forall k \in \K\\
  \quad & \Delta_k \geq 0 && \forall k \in \K\\ 
 \quad &  \eta \geq \zeta_{\mymbox{lb}}\\
\quad &   \eta \leq \zeta_{\mymbox{ub}} \label{sub_eta}\\
\quad & \BFone^T \BFell \leq B \label{budget}\\
\quad &  \BFell \in \{0,1\}^{\NE},
\end{align}
\end{subequations}
where $\BFlambda_k \in \Re^{\NE},\; \BFDelta \in \Re^{\NK},\eta \in \Re, \; \BFupsilon_k \in \Re^{\lvert V \rvert}$ and $\BFUpsilon \in \Re^{\NE\times\NK}$ while $\BFvarphi^*,\; p^*, \pi^*$ are the elements of $\BFpsi^*$ associated with \eqref{eq:DRIP:RC:C1:n}, \eqref{eq:DRO_u_sum_one:relaxed}, \eqref{RRLT:relaxed} respectively.
\end{proposition}

\begin{proof}
See Appendix \ref{proof:thm:violatedConstReform}.
\end{proof}

For completeness, we present the pseudocode of the column generation algorithm described above in Algorithm \ref{alg:candcg}. Note that the algorithm is guaranteed to converge in a finite number of iterations given that at each iteration, either the algorithm terminates or a new element $\BFell\in\LL$ is added to $\hat{\LL}$, yet $|\LL|$ is finite.

\begin{algorithm}
\caption{Column Generation Algorithm to solve problem \eqref{RPM}}\label{alg:candcg}
\begin{algorithmic}[1]
\Procedure{ColumnGeneration}{}
\State $\hat{\LL}\gets \{\BFzero\}$
\While{$\hat{\LL}\neq \LL$}
\State Solve problem \eqref{RPM} with $\eta_\BFell=0$,$u_\BFell=0$, $\BFDelta_\BFell=\BFzero$ for all $\BFell\in\LL/\hat{\LL}$\label{alg:sbandbAlgo:StepX}
\State Identify a set of optimal dual variables $\BFvarphi^*$, $p^*$, and $\pi^*$ as defined in Proposition \ref{thm:violatedConstReform}.
\State Solve problem \eqref{eq:violatedConstMILP} to obtain optimal value $v^*$ and optimal $\BFell^*$
\If{$v^*\geq 0$}
\State \textbf{return} Optimal solution obtained in step \ref{alg:sbandbAlgo:StepX}
\Else
\State $\hat{\LL}\gets \hat{\LL} \cup \{\BFell^*\}$
\EndIf
\EndWhile
\State \textbf{return} Solve \eqref{RPM} and return optimal solution
\EndProcedure
\end{algorithmic}
\end{algorithm}

%%%%%%-----
%%%%%%----

\subsection{Using coordinate descent for  \texorpdfstring{$g_{\mymbox{ub}}(\I)$}{lg}}\label{subsec:upp}

Given an interval $\I\subseteq \bar{\I}$ we are looking for an upper bound on $g(\I)$ and a value of $\zeta\in\I$ such that $g(\I)$ matches this upper bound. In order to accomplish this task, we will first look back at the solution of problem \eqref{RPM} to identify the optimal support $\LL_{\mymbox{lb}}^*$ and distribution $\BFu_{\mymbox{lb}}^*$ of the lower bounding problem. We then perform coordinate descent on problem \eqref{eq:DRIP:RC} where $\LL$ is replaced with $\LL_{\mymbox{lb}}^*$ iterating between a step where $\BFu$ stays fixed at the best solution found so far, initially at $\BFu_{\mymbox{lb}}^*$, and a step where it is rather $\zeta$ that stays fixed. In both cases, the problem reduces to a linear program whose size is linear in the size of $\LL_{\mymbox{lb}}^*$. This procedure is considered to have converged when the relative improvement on optimal value is considered small enough. For completeness, we provide the pseudocode for the coordinate descent algorithm in Algorithm \ref{alg:CD}.

 \begin{algorithm}
\caption{Coordinate Descent Algorithm to obtain upper bound on problem \eqref{eq:DRIP:RC}}\label{alg:CD}
\begin{algorithmic}[1]
\Procedure{CoordinateDescent}{$\epsilon$, $\LL_{\mymbox{lb}}^*$, $\BFu_{\mymbox{lb}}^*$}
\State $\bar{\BFu}^* \gets \BFu_{\mymbox{lb}}^*$,  $\LL\gets \LL_{\mymbox{lb}}^*$
\Do
\State Solve problem \eqref{eq:DRIP:RC} with constraint $\BFu=\bar{\BFu}^*$ to get optimal value $t_1^*$ and optimal $\zeta_{\mymbox{ub}}^*(\I)$
\State Solve problem \eqref{eq:DRIP:RC} with constraint $\zeta=\zeta_{\mymbox{ub}}^*(\I)$ to get optimal value $t_2^*$ and optimal $\bar{\BFu}^*$
\doWhile{$t_2^*< (1-\epsilon)t_1^*$}%\lvert t_1^*\rvert>\lvert(1+\epsilon) t_2^*\rvert$}
\State $g_{\mymbox{ub}}(\I)\gets t_1^*$
\State \textbf{return} $g_{\mymbox{ub}}(\I)$,  $\zeta_{\mymbox{ub}}^*(\I)$
\EndProcedure
\end{algorithmic}
\end{algorithm}

\section{Numerical experiments}\label{sec:simulate}

We performed a series of  numerical experiments to show the convergence and numerical efficiency of the spatial branch and bound algorithm (in Section \ref{sec:numExp:converge}), and to compare the in-sample and out-of-sample performance of optimal randomized and deterministic interdiction plan strategies in Section \ref{sec:numExp:outSample}.   All algorithms were implemented in  Matlab 2019b using the YALMIP toolbox  and CPLEX 12.9.0 was used to solve all continuous and mixed-integer linear programs. The DRNI problem instances were generated using the network structure given in Figure \ref{fig:network_simulate}. Within the same column, the arcs can point upward or downward with equal probability whereas between different columns, the arcs always point in the direction of the sink. The number of rows is denoted by $m$ and the number of columns is denoted by $n$. This class of network instances is the same as the one used in \cite{Cormican1998}, \cite{Janjarassuk2008} and \cite{Atamturk2017} except that we allow any arc to be interdicted whereas they have assumed that the arcs  in the first and last column and those leaving the source node (s) or entering the sink (t) are not interdictable and have infinite capacity. 
%The direction of the arcs within a column is randomly drawn with equal chances for each direction. 
For each DRNI problem instance, capacity vector scenarios are drawn from a factor model, $\BFc:=F\BF{\xi}$, with each $\xi_i$ independently distributed according to an exponential distribution with mean $\mu_i$, for some fixed $F\in\Re^{\NE\times 2}$ and $\BF{\mu}\in\Re^2$ that were randomly generated for the given instance.
%ERICK : here is the text if we keep the old results
%for some fixed randomly selected $F\in\Re^{\NE\times 2}$ and $\BF{\mu}\in\Re^2$.

\begin{figure}[t]
\centering
\begin{tikzpicture}[
      mycircle/.style={
         circle, 
         draw=black,
         fill=white,
         fill opacity = 0.3,
         text opacity=1,
         inner sep=0pt,
         minimum size=10pt,
         font=\small},
      myarrow/.style={-Stealth},
         scale=0.4, mycircles/.style={
         circle,
         draw=black,
         fill=white,
         fill opacity = 0.5,
         text opacity=1,
         inner sep=0pt,
         minimum size=15pt,
         font=\small},
      myarrow/.style={-Stealth},
    %   node distance=2.4cm and 2.4cm
      ]
\node[mycircles] (c0) at (0, 0) {$s$};
\node[mycircle] (c1) at (3,3) {};
\node[mycircle] (c2) at (3,1) {};
\node[mycircle] (c3) at (3,-3) {};
\node[mycircle] (c4) at (5,3) {};
 \node[mycircle] (c5) at (5,1) {};
  \node[mycircle] (c6)  at (5,-3) {};
\node[mycircle] (c7) at (9,3){};
\node[mycircle] (c8) at (9,1){};
\node[mycircle] (c9) at (9,-3) {};
\node[mycircle] (c11) at (11,3){};
\node[mycircle] (c12) at (11,1){};
\node[mycircle] (c13) at (11,-3) {};
\node[mycircles] (c10) at (14,0) {$t$};
     \foreach \i/\j/\txt/\p in {% start node/end node/text/position
      c0/c1/ /above,
      c0/c2/ /above,
      c0/c3/ /right,
       c1/c2/ /right,
    c1/c4/ /,
      c2/c5/ /,
       c5/c4/ /,
      c3/c6/ /,
       c12/c11/ /right,
      c11/c10/ /right,
      c12/c10/ /above,
      c7/c11/ /,
      c8/c12/ /,
       c7/c8/ /,
      c9/c13/ /,
      c13/c10/ /right}
 \draw [myarrow] (\i) -- node[font=\small,\p] {\txt} (\j);
  \draw [myarrow] (c2) --node[midway,right]{ } (3,-0.6);
   \draw [myarrow]  (c3)--node[midway,right]{ }(3,-1.5);
     \draw [myarrow] (c5) -- (5,-0.6);
   \draw [myarrow]  (5,-1.5) --(c6);
     \draw [myarrow] (c8) -- (9,-0.6);
   \draw [myarrow]  (9,-1.5) --(c9);
 \draw [myarrow] (c4) -- (6.6,3);
  \draw [myarrow]  (7.5,3)--(c7);
   \draw [myarrow] (c5) -- (6.6,1);
  \draw [myarrow]  (7.5,1)--(c8);
     \draw [myarrow] (c6) -- (6.6,-3);
  \draw [myarrow]  (7.5,-3)--(c9);
  \draw [myarrow]  (11,-0.6)--node[midway,right]{ }(c12);
  \draw [myarrow] (11,-1.5)--node[midway,right]{ }(c13);
   \draw [myarrow]  (c0)--node[midway,above]{ }(3,-1);
    \draw [myarrow] (11,-1)--node[midway,above]{ }(c10);
  \foreach \x in {6.9,7.1,7.3} {
 \fill[color=black] (\x,3) circle (0.05);}
 \foreach \x in {6.9,7.1,7.3} {
  \fill[color=black] (\x,1) circle (0.05);}
      \foreach \x in {6.9,7.1,7.3}{
  \fill[color=black] (\x,-3) circle (0.05);}
 \foreach \y in {-0.9,-1.1,-1.3} {
     \fill[color=black] (3,\y) circle (0.05);}
     \foreach \y in {-0.9,-1.1,-1.3} {
    \fill[color=black] (5,\y) circle (0.05);}
  \foreach \y in {-0.9,-1.1,-1.3} {
  \fill[color=black] (9,\y) circle (0.05);}
     \foreach \y in {-0.9,-1.1,-1.3} {
   \fill[color=black] (11,\y) circle (0.05);}
    \foreach \x in {6.9,7.1,7.3} {
   \fill[color=black] (\x,-1*\x+6) circle (0.05);}
    \end{tikzpicture}
    \caption{Network for numerical experiments}
    \label{fig:network_simulate}
    \end{figure}
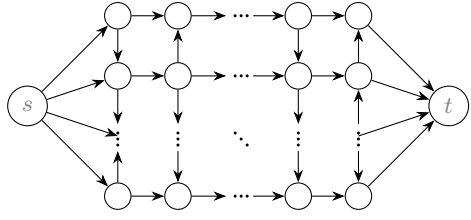

\subsection{Convergence of spatial branch and bound algorithm}\label{sec:numExp:converge}
In this subsection, we show that the spatial branch and bound algorithm described in Section \ref{sec:sBB} converges quickly to approximately optimal solutions. We randomly generate $10$ instances of sizes $10 \times 10$, $20\times 20$ and $30\times30$ of the DRNI problem as described in the introduction of this section. For each instance, we consider $|K|=20$ randomly generated scenarios from the constructed factor model.
%Before conducting the experiments on the network of a given size, we generate $10000$ scenarios for the capacity of the network and uniformly sample $20$ scenarios at random from them. 
We then solve the instances for a convergence tolerance, $\epsilon$, either equal to $0.01\%$ or $0.0001\%$, terminating the algorithm if the higher precision is not achieved after $1$ hour.  The results of the numerical experiments for $\Gamma=2$, $B=5$, $\hat{\BFq}=(1/20)\BFone$,
$\bar{\BFq}=\BFone$,  $\alpha=0.05$ are given in Table \ref{tab:DR interdiction}. It can be seen that for networks of sizes $10 \times 10$ and $20 \times 20$ , the spatial branch and bound algorithm converges within $1$ hour to an accuracy of $\epsilon=0.0001\%$ for $9$ problem instances (out of $10$) whereas we found that for networks of size $30\times30$ the algorithm did not reach this level of accuracy within $1$ hr for $2$ instances out of $10$. On the other hand, for sizes of problem $10 \times 10$ and $20 \times 20$, all instances converged within less than $1$ hour to a precision of $\epsilon=0.01\%$, reaching on average a gap close to $0.007\%$ and $0.001\%$ respectively in the case of the problematic instances. However, for one instance of problem $30\times30$, the algorithm did not converge in $1$ hr even for $0.01\%$ precision, and the gap was found to be close to $0.25\%.$ % In the numerical experiments, we found that in the first few iterations, the algorithm identifies the upper bound which is nearly optimal. Therefore, we studied the approximate solution for the
\begin{table}[htbp]
\caption
{Convergence of spatial branch and bound algorithm for 10 randomly generated instances. * avg. cpu time is reported for those instances which converged with $.0001\%$ precision in $1$ hour}\label{tab:DR interdiction}
\begin{tabular}{@{}ccccc@{}}
\toprule
Problem size $(m\times n)$
    %   & \multicolumn{2}{c}{$\epsilon=10^{-6}$} 
    %       & \multicolumn{2}{c}{$\epsilon=10^{-4}$} \\           \cline{2-3} \cline{4-5}

  & \multicolumn{2}{c}{\centering avg. cpu time (min)$^*$} & \multicolumn{2}{c}{\centering gap $>0.0001\%$ after $1$ hr}\\
  \cmidrule{2-3} \cmidrule{4-5}
  & $\epsilon=0.0001\%$&$\epsilon=0.01\%$ & $\#$ instances & optimality gap
  \\
      \midrule
  $10\times10$, $\NE=200$ & $2.86$ &$1.85$&$1$ & $0.007\%$ \\ 
  $20\times20$, $\NE=800$ & $8.32$ & $3.13$ &$1$ &$0.001\%$ \\
  $30\times30$, $\NE=1800$ & $7.80$  &$6.44$  &$2$&  $0.14\%$ \\
\bottomrule
\end{tabular}
 \end{table}

\subsection{In-sample and out-of-sample performance of randomized and deterministic strategies} \label{sec:numExp:outSample}

To illustrate the benefit of randomization, we focus our attention on the network given in Figure \ref{ex:out_sample} with $4$ rows and  $2$ columns while the factor model continues to be regenerated for each instance.
We are interested in comparing the in-sample and out-of-sample performance of strategies that are obtained using only $|\K|=20$ scenarios from the  underlying distribution. To do so, we generated $100$ set of samples for $\{\BFc^1,\dots,\BFc^{20}\}$ for each level of $\Gamma  \in \{0,~0.1,~0.5, ~1,~ 10, ~20\}$.  For each set of samples, we compare the in-sample performance of both the randomized interdiction strategy $\hat{\BFu}_{20}^*$ and the deterministic strategy $\hat{\ell}_{20}^*$ that optimizes the DRNI problem constructed from this \quoteIt{observed} sample set, with $B=1$, $\alpha=0.05$, $\hat{\BFq}=(1/20)\BFone$, and $\bar{\BFq}=\BFone$. 
We then compare the performance of the in-sample optimal strategies $\hat{\BFu}_{20}^*$ and $\hat{\BFell}_{20}^*$ on the \quoteIt{unobserved} underlying distribution using a Monte-Carlo simulation of $100000$ scenarios. The in-sample optimal randomized strategy $\hat{\BFu}_{20}^*$ is obtained by using the spatial branch and bound algorithm outlined in Section \ref{sec:sBB} with $\epsilon <0.01\%$, while the optimal in-sample deterministic strategy $\hat{\BFell}_{20}^*$ is computed  by solving the MILP given in Appendix \ref{appen:deterministic}. 

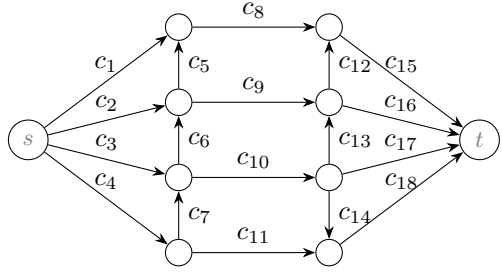
\begin{figure}[htbp]
\centering
\begin{tikzpicture}[
      mycircle/.style={
         circle,
         draw=black,
         fill=white,
         fill opacity = 0.3,
         text opacity=1,
         inner sep=0pt,
         minimum size=10pt,
         font=\small},
      myarrow/.style={-Stealth},
     mycircles/.style={
         circle,
         draw=black,
         fill=white,
         fill opacity = 0.5,
         text opacity=1,
         inner sep=0pt,
         minimum size=15pt,
         font=\small},
      myarrow/.style={-Stealth},
    %   node distance=2.4cm and 2.4cm
      ]
\node[mycircles] (c0) at (0, 0) {$s$};
\node[mycircle] (c1) at (2,1.5) {};
\node[mycircle] (c2) at (2,0.5) {};
\node[mycircle] (c3) at (2,-0.5) {};
\node[mycircle] (c4) at (2,-1.5) {};
\node[mycircle] (c5) at (4,1.5) {};
\node[mycircle] (c6) at (4,0.5) {};
\node[mycircle] (c7) at (4,-0.5) {};
\node[mycircle] (c8) at (4,-1.5) {};
\node[mycircles] (c10) at (6,0) {$t$};
\path[myarrow]
(c0)  edge  node[above]  {$c_1$}         (c1)
    (c0)  edge   node[above]  {$c_2$}         (c2)
    (c0)  edge  node[above]  {$c_3$}         (c3)
     (c0)  edge  node[above]  {$c_4$}         (c4)
     (c1)  edge   node[above]  {$c_8$}         (c5)
     (c2)  edge  node[right]  {$c_{5}$}         (c1)
      (c3)  edge  node[right]  {$c_{6}$}         (c2)
       (c4)  edge  node[right]  {$c_{7}$}         (c3)
     (c2)  edge  node[above]  {$c_{9}$}         (c6)
      (c3)  edge   node[above]  {$c_{10}$}         (c7)
     (c4)  edge  node[above]  {$c_{11}$}         (c8)
      (c6)  edge  node[right]  {$c_{12}$}         (c5)
       (c7)  edge   node[right]  {$c_{13}$}         (c6)
      (c7)  edge  node[right]  {$c_{14}$}         (c8)
       (c5)  edge   node[above]  {$c_{15}$}         (c10)
     (c6)  edge  node[above]  {$c_{16}$}         (c10)
     (c7)  edge  node[above]  {$c_{17}$}         (c10)
      (c8)  edge  node[above]  {$c_{18}$}         (c10);
    %  \foreach \i/\j/\txt/\p in {% start node/end node/text/position
    %   c0/c1/ /above,
    %   c0/c2/ /above,
    %   c0/c3/ /right,
    %       c0/c4/ /right,
    %             c0/c1/ /above,
    %           c2/c1/ /above,
    %             c3/c2/ /above,      c4/c3/ /above,   
    %   c1/c5/ /above,
    %   c2/c6/ /right,
    %       c3/c7/ /right,
    %       c4/c8/ /right,           c6/c5/ /above,
    %             c7/c6/ /above,      c7/c8/ /above,
    %           c5/c10/ /above,
    %   c6/c10/ /right,
    %       c7/c10/ /right,
    %       c8/c10/ /right}\draw [myarrow] (\i) -- node[font=\small,\p] {\txt} (\j);
    \end{tikzpicture}
    \caption{Example network}
    \label{ex:out_sample}
    \end{figure}

In terms of in-sample performance, we are interested in comparing the added value of the randomized strategy, which can be defined as the relative gap between the worst-case CVaR obtained by the deterministic  and randomized strategies:
    %\begin{align*}
 %\mbox{VRS}=\frac{\mymbox{WC-CVaR}_{\hat{\BFell}^*}-\mymbox{W-CVaR}_{\hat{\BFu}^*}}{\mymbox{W-CVaR}_{\hat{\BFell}^*}}\times 100\%.
 %   \end{align*}
     \begin{align*}
 \mbox{VRS}=\frac{\max_{\BFq\in \hat{\Q}_{20}} \mbox{CVaR}^\alpha_{ k\sim \BFq}[ f_{\hat{\BFell}_{20}^*, k}]-\max_{\BFq\in \hat{\Q}_{20}} \mbox{CVaR}^\alpha_{ \BFell\sim \hat{\BFu}_{20}^*,k\sim \BFq}[ f_{\BFell, k}]}{\max_{\BFq\in \hat{\Q}_{20}} \mbox{CVaR}^\alpha_{ \BFell\sim \hat{\BFu}_{20}^*,k\sim \BFq}[ f_{\BFell, k}]}\times 100\%.
    \end{align*}
where $\hat{\Q}_{20}$ is centered at the empirical distribution over the sample set. Table \ref{tab:in-sample} reports the number of instances, for each level of $\Gamma$, for which VRS was equal to zero, between zero and 1\%, and above $1\%$. For the later case, the average VRS is also reported. We find that, for $\Gamma \in \{0,~0.1\}$, all $100$ instances achieved a VRS below $1\%$. This is reasonable given that, as $\Gamma$ goes to $0$, the DRNI problem converges to CVaR minimization with known distribution for which it is known that deterministic strategies are optimal. We also observe that as distributional ambiguity is increased, the randomized strategies perform significantly  better than deterministic ones with respect to the in-sample DRNI problem instance, reaching for $\Gamma=20$ an average VRS of $4.67\%$.

\begin{table}[htbp]
\caption{In-sample value of randomized solution performances } \label{tab:in-sample}
\begin{tabular}{@{}ccccc@{}}
\toprule
$\Gamma$ &  \multicolumn{1}{p{3cm}}{\centering $\mbox{VRS}=0$}&  \multicolumn{1}{p{3cm}}{\centering $0<\mbox{VRS}<1\%$} &\multicolumn{2}{c}{\mbox{VRS} $\geq1 \%$}\\
\cmidrule{4-5}
& $\#$ instances& $\#$ instances  &  $\#$ instances & avg. \mbox{VRS}  \\ \midrule
$0$ & $100$ & $0$ & $0$& $0$\\
$0.1$ & $98$ &$2$ & $0$& $0$\\
$0.5$ & $94$ &$4$ &$2$ &$2.09\%$ \\
$1$ & $89$ & $5$ & $6$ & $3.48\%$ \\
$10$ & $78$ &$12$ &$10$   &$4.45\%$\\
$20$ & $78$ & $14$ & $8$  & $4.67\%$ \\
\bottomrule
% $0$ & $0$ & $0$& $0$\\
% $0.25$& $0$&$0$& $0$\\
% $0.5$   & $3$& $0$& $<1$\\
% $1$ & $9$& $6$& $3.2563$ \\
% $2$     & $13$& $12$& $3.2290$\\
% $10$    & $12$& $8$& $3.5497$\\
% $20$    & $10$& $6$& $3.9615$\\
\end{tabular}
\end{table}

\begin{table}[htbp]
\caption{Comparison of out-of-sample CVaR performance for deterministic and randomized strategy for instances where  randomized strategy improved on the in-sample performance of deterministic strategy by at least $1\%$.}\label{tab:out-sample}
 \begin{tabular}{@{}cccccc@{}}
\toprule
$\Gamma$ &  $\#$ of instances &
avg. $\mbox{CVaR}_{r}$ &
avg. $\mbox{CVaR}_{d}$ & \multicolumn{1}{p{3cm}}{\centering $\#$ of instances $\mbox{CVaR}_r<\mbox{CVaR}_d$}&$(\mbox{CVaR}_r-\mbox{CVaR}_d)/\mbox{CVaR}_d$ \\ \midrule
$0.5$ & $2$& $1.09$ & $1.16$ & $2$&$36\%$\\
$1$ & $6$&$0.84$ & $1.07$ & $6$&$19\%$\\
$10$ & $10$&$0.77$ &$0.93$ & $10$&$17\%$\\
$20$ &$8$ &$0.64$ &$0.78$ & $8$ &$19\%$\\
\bottomrule
\end{tabular}
 \end{table}

In terms of out-of-sample performance, we only consider problem instances for which the VRS was above 1\%, given that otherwise the optimal in-sample randomized and deterministic strategies are too similar. We denote the optimal out-of-sample CVaR corresponding to randomized strategies and deterministic interdiction plans by $\mbox{CVaR}_{r}$ and $\mbox{CVaR}_{d}$ respectively. It can be clearly seen from Table \ref{tab:out-sample} that  randomized strategies have lower average CVaR than deterministic interdiction plans. In all the out-of-sample instances, we found that the optimal randomized strategies always have lower CVaR than optimal deterministic interdiction plans.  
%The distribution of the average out-of-sample CVaR for randomized and deterministic strategies  for $\Gamma=0.5,\, 1,\, 10,\, 20$ is also provided for completeness in Figure \ref{Gamma1}. 
We argue that this evidence confirms that there is a real observable benefit for the network interdictor for employing randomized interdiction plans in a risk averse network interdiction game, both in a distributionally robust setting (c.f. the in-sample performance comparison) and in a setting where the network capacity distribution information comes from a limited number of observed realizations.

\section{Conclusions}\label{sec:Conclusions}
 In this paper, we introduced a distributionally robust risk averse network interdiction game to model the strategic interactions between a risk-averse interdictor and the flow player.  We solved the interdictor's non-linear bi-convex DRO problem by first reformulating it as a bi-convex optimization problem using LP duality and then devising a spatial branch and bound algorithm. After observing that the optimal randomized strategy  can be supported on a small number of interdiction plans in DRO problems, we developed a column generation algorithm that can be used to efficiently determine the convex relaxation of the interdictor's problem. Our numerical experiments showed that 1) our proposed spatial branch and bound algorithm can efficiently solve distributionally robust interdiction games of reasonable sizes; 2) randomization can be quite effective in reducing the risk exposure obtained from the optimal deterministic interdiction strategy both when comparing in-sample worst-case CVaR and out-of-sample CVaR performances. 

Given that Stackelberg games with single leader and multiple followers have been extensively applied in the literature, it would be interesting as future work to extend our algorithm in a way that can address Stackelberg games with a leader that is both risk and ambiguity averse while  followers implement a Nash equilibrium that accounts for their respective risk aversion. 

% \bibliography{references.bib}
\bibliographystyle{apalike}

  \appendix
 \section{Proofs}

%  %%%%-----
%  %%%%%-----
%  %%-----
  \subsection{ Proof of Proposition \ref{prop:interdictorDRO}}\label{proof:prop:interdictorDRO}
%  \begin{proof}
 We substitute the expression of CVaR in \eqref{Network_equilibrium} to obtain
\begin{align}
 \mathop{\minimize_{\BFu \in \Delta \LL} } \max\limits_{\BFq \in \Q} \inf_\zeta \quad & \zeta + \frac{1}{1-\alpha}\sum\limits_{\BFell}\sum\limits_{k}q_k u_\BFell[f_{\BFell, k}- \zeta]^+. \label{single_OP}
\end{align}
We prove by contradiction that an optimal $\zeta$, denoted by $\zeta^*$, lies between $0$ and $\bar{\zeta}$.  First, we assume that $\zeta^*=\zeta_L<0$ is the largest optimal value for $\zeta$. The maximum  flow $f_{\BFell, k}$ for any $\BFell$ and $k$ is bounded below by 0, hence for any $\BFq\in\Q$ the CVaR equals
\begin{align*}
    \zeta_L + \frac{1}{1-\alpha}\sum\limits_{\BFell}\sum\limits_{k}q_k u_\BFell(f_{\BFell, k}- \zeta_L)= \frac{(1-\alpha)\zeta_L-\zeta_L}{1-\alpha} +\frac{1}{1-\alpha}\sum\limits_{\BFell}\sum\limits_{k}q_k u_\BFell f_{\BFell, k}.
\end{align*}
However, we arrive at a contradiction since $\zeta=0$ is at least as good as $\zeta_L$:
\begin{align*}
    \frac{-\alpha \zeta_L }{1-\alpha}+ \frac{1}{1-\alpha}\sum\limits_{\BFell}\sum\limits_{k}q_k u_\BFell f_{\BFell, k} \geq \frac{1}{1-\alpha}\sum\limits_{\BFell}\sum\limits_{k}q_k u_\BFell  f_{\BFell, k} = 0 + \frac{1}{1-\alpha}\sum\limits_{\BFell}\sum\limits_{k}q_k u_\BFell[f_{\BFell, k}- 0]^+,
\end{align*}
for any $\alpha \in [0,1)$. So, we can conclude that $\zeta^*:=0$ is also optimal.

Alternatively, we can assume that  $\zeta^*=\zeta_H>\bar{\zeta}$ is the smallest optimal solution for $\zeta$. In this case, for any $\BFq\in\Q$ we have that CVaR equals:
\begin{align*}
    \zeta_H + \frac{1}{1-\alpha}\sum\limits_{\BFell}\sum\limits_{k}q_k u_\BFell [f_{\BFell,k}-\zeta_H]^+ =\zeta_H ,
\end{align*}
since $f_{\BFell,k}\leq f_{\BFzero,k} \leq \bar{\zeta} < \zeta_H$. 
Yet, for $\zeta=\bar{\zeta}$, the worst-case CVaR, given by $\bar{\zeta}$, is strictly less than $\zeta_H$ which contradicts our assumption that $\zeta_H$ is the smallest optimal solution for $\zeta$.

We now turn ourselves to establishing the reformulation presented as problem \eqref{epi}. Let $v(\zeta, \BFq, \BFu)= \zeta + \frac{1}{1-\alpha}\sum\limits_{\BFell}\sum\limits_{k}q_k u_\BFell[f_{\BFell, k}- \zeta]^+$. Since $v(\zeta,\BFq,\BFu)$ is convex in $\zeta$ for all $\BFq \in \Q$ while being linear in $\BFq$ for all $\zeta\in[0,\,\bar{\zeta}]$, and since $[0,\,\bar{\zeta}]$ is bounded and $\Q$ is convex, it follows from Sion's minimax theorem (see \cite{Sion1958}) that \eqref{single_OP} is equivalent to 
  \begin{align}
 \mathop{\minimize_{\BFu \in \Delta \LL} }\min_{0\leq\zeta\leq \bar{\zeta}} \max\limits_{\BFq \in \Q}\quad & \zeta + \frac{1}{1-\alpha}\sum\limits_{\BFell}\sum\limits_{k}q_k u_\BFell[f_{\BFell, k}- \zeta]^+. \label{Sion_OP}
\end{align}
Since we are minimizing over $\BFu$ and $\zeta$, we have an equivalent reformulation of \eqref{Sion_OP} given by 
%\begin{subequations}
  \begin{align*}
  \mathop{\minimize_{\BFu \in \Delta\LL,\BFDelta\geq 0,}}_{0\leq\zeta\leq \bar{\zeta}} \max_{\BFq \in \Q}  \quad & \zeta + \frac{1}{1-\alpha}\sum\limits_{\BFell}\sum\limits_{k}q_{k} \Delta_{\BFell,  k}\\
\subto  \quad &  \Delta_{\BFell,  k} \geq u_\BFell f_{\BFell, k}-u_\BFell\zeta && \forall \BFell  \in \LL,\, k \in \K.
  \end{align*}
  %\end{subequations}
Employing an epigraph representation of the above objective function and introducing the decision variable $\eta_\BFell=u_\BFell \zeta$, we obtain problem \eqref{epi}.
%\end{proof}

\subsection{Proof of Proposition \ref{theorem:robust_counterpart}}\label{append:proof_robust_counterpart}
%\begin{proof}
Substituting $\BFq=\hat{\BFq}+\diag(\bar{\BFq})\BFz$, we obtain that \eqref{eq:DRO_non_linear_constraint} is equivalent to
\begin{align}
\max_{\BFz:\BFz\in\mathcal{Z}(\Gamma),\sum_{k\in\K}\left(\hat{q}_k+\bar{q}_kz_k\right) = 1,\hat{q}_k+\bar{q}_kz_k \geq 0,\forall k\in\K}    \zeta + \frac{1}{1-\alpha}\sum\limits_{\BFell \in \LL}\sum\limits_{k \in \K} (\hat{q_k}+\bar{q}_kz_k) \Delta_{\BFell,  k}-t\leq 0,\label{eq:proof:robustReform1}
\end{align}
where we will exploit a well-known equivalent representation of $\mathcal{Z}(\Gamma)$:
\[\mathcal{Z}(\Gamma)=\left\{\BFz\in\Re^{|\K|}\,\middle|\,\exists\BFdelta^+,\BFdelta^-\in\Re^{|\K|},\;\begin{array}{l}\BFz=\BFdelta^+-\BFdelta^-\\0\leq \BFdelta^+ \leq 1 \\
    0 \leq \BFdelta^- \leq 1\\
    \sum_{k\in\K} \delta_k^+ +  \delta_k^- \leq \Gamma\end{array}\right\}.\]
    Let us define $\gamma_k :=\frac{1}{1-\alpha} \sum_{\BFell \in \LL}\Delta_{\BFell,  k}$, so that we can rewrite the maximization problem in \eqref{eq:proof:robustReform1}  as follows
    \begin{subequations}
\begin{align}
    \mathop{\maximize_{\BFdelta^+,\BFdelta^-}} \quad & \zeta+\sum_{k\in \K}\left(\gamma_k(\hat{q}_k+\bar{q}_k(\delta_k^+-\delta_k^-)\right)-t\\
    \subto \quad & 0\leq \BFdelta^+ \leq 1 \label{eq:proof:RC:max:C1}\\
    \quad & 0 \leq \BFdelta^- \leq 1 \label{eq:proof:RC:max:C2}\\
    \quad & \sum_{k\in\K} \delta_k^+ +  \delta_k^- \leq \Gamma\label{eq:proof:RC:max:C3}\\
    \quad & \sum_{k\in\K} \bar{q}_k\left(\delta_k^+ - \delta_k^-\right)=0\label{eq:proof:RC:max:C4}\\
    \quad & \hat{q}_k +\bar{q}_k\delta_k^+ -\bar{q}_k\delta_k^- \geq 0 &&\forall k \in \K.\label{eq:proof:RC:max:C5}
\end{align}
\end{subequations}

Since $\BFdelta^+=\BFdelta^-=0$ is feasible, we can use strong LP duality to obtain an equivalent problem:
\begin{align*}
\minimize_{\scriptsize \begin{array}{c}\BFw,\BFw^-,\\\chi,\bar{\mu},\bar{\BFbeta}\end{array}}\quad & \zeta+\sum_{k\in \K}w_k+\sum_{k\in \K}w_k^-+\Gamma\chi +\sum_{k\in\K}\left((\gamma_k+\bar{\beta}_k)\hat{q}_k\right)-t\\
   \subto \quad & \bar{q}_k\gamma_k-w_k-\chi+\bar{q}_k\bar{\mu}+\bar{q}_k\bar{\beta}_k\leq 0 && \forall k \in \K\\
    \quad & -\bar{q}_k\gamma_k-w_k^--\chi-\bar{q}_k\bar{\mu}-\bar{q}_k\bar{\beta}_k\leq 0 && \forall k \in \K\\
 \quad & \BFw\geq 0,\BFw^-\geq 0,\chi \geq 0,\bar{\BFbeta}\geq 0,  
\end{align*}
where $\BFw\in\Re^{|\K|}$, $\BFw^-\in\Re^{|\K|}$, $\chi\in\Re$, $\bar{\mu}\in\Re$, and $\bar{\BFbeta}\in\Re^{|\K|}$ are the dual variables associated to constraints \eqref{eq:proof:RC:max:C1} to \eqref{eq:proof:RC:max:C5} respectively.
Substituting $\BFbeta:=\BFgamma+\bar{\mu}+\bar{\BFbeta}$ and then $\mu:=-\bar{\mu}$, we obtain 
\begin{subequations}
\begin{align}
    \minimize_{\scriptsize \begin{array}{c}\BFw,\BFw^-,\\\chi,\mu,\BFbeta\end{array}}\quad & \zeta+\sum_{k\in \K}w_k+\sum_{k\in \K}w_k^-+\Gamma\chi +\sum_{k\in \K}\hat{q}_k\beta_k+\mu-t\\
     \subto  \quad &\mu\geq \gamma_k-\beta_k &&\forall k 
    \in \K \label{eq:mu:p}\\
    \quad & \chi\geq  \bar{q}_k\beta_k-w_k&& \forall k \in \K \label{eq:chi_plus:p}\\
    \quad & \chi\geq -\bar{q}_k\beta_k-w^-_k&&\forall k 
    \in \K \label{eq:chi_minus:p}\\
 \quad & \BFw\geq 0,\BFw^-\geq 0,\chi \geq 0\label{eq:w:chi:p}.
\end{align}
\end{subequations}

Combining the above problem with \eqref{epi}, we obtain
  %\begin{subequations}
  \begin{align*}
 \minimize_{\scriptsize \begin{array}{c}\BFu, t,\BFw,\BFw^-,\\\BFeta,\chi, \BFbeta, \Delta, \zeta\end{array}} \quad & t\\ 
\subto \quad &  \eqref{eq:DRO_Delta}-\eqref{eq:DRO_zeta_bound}, \eqref{eq:chi_plus:p}-\eqref{eq:w:chi:p} \\
 \quad & \zeta+\sum_{k' \in \K}w_{k'}+\sum_{k' \in \K}w_{k'}^-+\Gamma\chi \notag\\\qquad&+\sum_{k'\in\K} 
 \hat{q}_{k'}\beta_{k'}+\frac{1}{1-\alpha}\sum\limits_{\BFell \in \LL} \Delta_{\BFell,  k}-\beta_k \leq t&& \forall k\in\K \,,%\label{eq:DRIP:RC:C1:p}
  \end{align*}
  %\end{subequations}
where we were able to replace $\mu$ with $\max_{k\in\K} \frac{1}{1-\alpha}\sum\limits_{\BFell \in \LL} \Delta_{\BFell,  k}-\beta_k$.
%\end{proof}

 \subsection{Proof of Proposition \ref{thm:violatedConstReform}}\label{proof:thm:violatedConstReform}
%\begin{proof}
We start by repeating the definition of condition \eqref{eq:violatedConstraintCond}:
\begin{align*}
   \quad & \inf_{\BFell \in \LL }\sup_{ \BFsigma_{\BFell} \geq 0} \inf_{\BFy_{\BFell}}\; {\BFpsi^*}^T B_{\BFell}\BFy_{\BFell} + \BFsigma^T_{\BFell}W_{\BFell}\BFy_{\BFell}\geq 0.
\end{align*}
We will first argue that the order of $\sup_{ \BFsigma_{\BFell} \geq 0}$ and $\inf_{\BFy_{\BFell}}$ can be changed without affecting the value that is obtained. In particular,
\begin{align*}
\sup_{ \BFsigma_{\BFell} \geq 0} \inf_{\BFy_{\BFell}} \;{\BFpsi^*}^T B_{\BFell}\BFy_{\BFell} + \BFsigma^T_{\BFell}W_{\BFell}\BFy_{\BFell} = \sup_{\BFsigma_{\BFell} \geq 0:B_{\BFell}^T\BFpsi^*+W_{\BFell}^T\BFsigma_{\BFell}=0} 0,\end{align*}
while 
\begin{align*}
\inf_{\BFy_{\BFell}}\sup_{ \BFsigma_{\BFell} \geq 0} {\BFpsi^*}^T B_{\BFell}\BFy_{\BFell} + \BFsigma^T_{\BFell}W_{\BFell}\BFy_{\BFell} = \inf_{\BFy_{\BFell}:W_{\BFell}\BFy_{\BFell}\leq 0} {\BFpsi^*}^T B_{\BFell}\BFy_{\BFell}.\end{align*}
Since the two problems are dual of each other and $\inf_{\BFy_{\BFell}:W_{\BFell}\BFy_{\BFell}\leq 0} {\BFpsi^*}^T B_{\BFell}\BFy_{\BFell}$ is feasible with $\BFy_\BFell=0$, we can conclude by strong LP duality that the two values are the same.

We thus obtain that the left-hand side of condition \eqref{eq:violatedConstraintCond} can be obtained by solving:
\begin{align*}
 \minimize_{\BFell \in \LL ,\BFy_{\BFell}}\sup_{ \BFsigma_{\BFell} \geq 0 } {\BFpsi^*}^T B_{\BFell}\BFy_{\BFell} + \BFsigma^T_{\BFell}W_{\BFell}\BFy_{\BFell}.  
\end{align*}
On taking the dual of the inner maximization problem, we have
\begin{align*}
    \minimize_{\BFell  \in \LL,\BFy_{\BFell}} \quad & {\BFpsi^*}^T B_{\BFell}\BFy_{\BFell}  && \qquad \\
   \subto \quad & W_{\BFell}\BFy_{\BFell}\leq 0,
\end{align*}
which can be written more carefully as
\begin{subequations}\label{eq:proof:propIII:probA}
\begin{align}
   \mathop{\minimize_{\BFell \in \LL,\bar{\BFDelta}, \bar{u}, \bar{\eta}}} \quad & \frac{{\BFvarphi^*}^T}{1-\alpha} \bar{\BFDelta}+p^* \bar{u}+\pi^*
    \bar{\eta} \\
    \subto \quad &  \bar{\Delta}_{k} \geq \bar{u} f_{\BFell  k}^* -\bar{\eta} &&\forall k \in \K \\
  \quad & \bar{\Delta}_{k}\geq 0 &&\forall k \in \K \\
  \quad &\bar{u} \geq 0\\
 \quad &  \bar{\eta} \geq \bar{u} \zeta_{\mymbox{lb}}\\
\quad &   \bar{\eta} \leq \bar{u} \zeta_{\mymbox{ub}}\\
\quad &  \BFell \in \{0,1\}^{\NE}\\
\quad & \BFone^T \BFell=1,
\end{align}
\end{subequations}
where $\bar{\BFDelta} \in \Re^\NK, \bar{u}\in \Re, \bar{\eta} \in \Re$, and where $\BFvarphi^*$, $p^*$, and $\pi^*$ are the terms of the dual vector $\BFvarphi^*$ associated with constraints \eqref{eq:DRIP:RC:C1:n}, \eqref{eq:DRO_u_sum_one:relaxed}, and \eqref{RRLT:relaxed} respectively. 

Next, we can observe that when $\bar{u}=0$, problem \eqref{eq:proof:propIII:probA} necessarily evaluates to zero. From this we conclude that $\bar{u}>0$ can be added to problem \eqref{eq:proof:propIII:probA} without affecting the conclusion when used to check condition \eqref{eq:violatedConstraintCond}. Moreover, $\bar{u}$ can be pulled out of \eqref{eq:proof:propIII:probA} after replacing $\BFDelta := (1/\bar{u})\bar{\BFDelta}$ and $\eta := \bar{\eta}/\bar{u}$ to obtain:
\begin{subequations}\label{eq:proof:propIII:probB}
\begin{align}
   \mathop{\minimize_{\BFell \in \LL,\BFDelta, \eta}} \quad & \frac{{\BFvarphi^*}^T}{1-\alpha} \BFDelta+p^*+\pi^*
    \eta \\
    \subto \quad &  \Delta_{k} \geq f_{\BFell  k}^* -\eta &&\forall k \in \K \\
  \quad & \Delta_{k}\geq 0 &&\forall k \in \K \\
 \quad &  \eta \geq  \zeta_{\mymbox{lb}}\\
\quad &   \eta \leq \zeta_{\mymbox{ub}}\\
\quad &  \BFell \in \{0,1\}^{\NE}\\
\quad & \BFone^T \BFell=1,
\end{align}
\end{subequations}

Next, in order to solve the problem \eqref{eq:proof:propIII:probB}, we need to make explicit the relation between $\BFell$ and $f_{\BFell, k}$ for each scenario $k$. 
One way is to exploit the dual problem associated with problem \eqref{eq:Primal_flowP} which is given by 
 %\begin{subequations}\label{eq:Dual_flowP}
 \begin{align*}
     f_{\BFell, k}=\quad \min_{\BFupsilon,\BFlambda} \quad&(\BFone-\BFell)^TC^k\BFlambda  \\
     \quad \subto\quad& \BFlambda+N^T\BFupsilon -\BFd \geq 0 \\%\label{Dual_constr}\\
     &0\leq\BFlambda \leq 1,% \label{eq:bound:lambda}
 \end{align*}
%\end{subequations}
where $\BFupsilon \in \Re^{\lvert V \rvert}$ and $\BFlambda \in \Re^{\NE}$ are duals associated with the  constraints \eqref{eq:flow_balance_Incidence}
and \eqref{eq:PrimFlow_Capacity} respectively. 
The dual vector $\BFlambda$ is bounded above by $1$ since a unit increase in capacity of an arc can increase the flow by at most one unit, see, \cite[Lemma 1]{Cormican1998}. 

We therefore have reduced the evaluation of the left-hand side of condition \eqref{eq:violatedConstraintCond} to solving the following non-linear 
mixed integer programming  (NL-MIP) problem 
\begin{subequations}
\begin{align}
    \mathop{\minimize_{\BFell \in \LL, \BFDelta, \eta}}_{\{\BFlambda_k, \BFupsilon_k\}_{k\in\K}} \quad & \frac{{\BFvarphi^*}^T}{1-\alpha} \BFDelta+p^*+\pi^*
    \eta \\
    \subto \quad &  \Delta_k \geq  \BFone^TC^k\BFlambda_k - \BFell^TC^k\BFlambda_k -\eta && \forall k \in \K\\
  \quad & \BFlambda_k+N^T\BFupsilon_k -\BFd \geq 0 && \forall k \in \K \label{sub_lambda:proof}\\
   \quad & 0\leq\BFlambda_k \leq 1 && \forall k \in \K\\
  \quad & \Delta_k \geq 0 && \forall k \in \K\\ 
 \quad &  \eta \geq \zeta_{\mymbox{lb}}\\
\quad &   \eta \leq \zeta_{\mymbox{ub}} \\
\quad &  \BFell \in \{0,1\}^{\NE}\\
\quad & \BFone^T \BFell \leq B. \label{budget:proof}
\end{align}
\end{subequations}
The non-linearity in the above problem is due to the bilinear terms $\BFell^TC_k\BFlambda_k=\BFc_k^T\diag(\BFell)\BFlambda_k$. We can  linearize them  to obtain an equivalent MILP since $\BFell$ is binary:
%\begin{subequations}
\begin{align*}
    \mathop{\minimize_{\BFell \in \LL, \BFDelta, \eta}}_{\{\BFlambda_k, \BFupsilon_k,\BFUpsilon_k\}_{k\in\K}} \quad & \frac{{\BFvarphi^*}^T}{1-\alpha} \BFDelta+p^* +\pi^*\eta \\
    \subto \quad &  \Delta_{k}\geq \BFc_k^T\BFlambda_k- \BFc_k^T\BFUpsilon_k -\eta &&\forall k \in \K\\
  \quad &  \BFUpsilon_k \leq \BFell &&k \in \K\\
  \quad &  \BFUpsilon_k \leq \BFlambda_k &&\forall k \in \K\\
   \quad & \BFUpsilon_k \geq \BFlambda_k+\BFell -1&&\forall k \in \K\\
  \quad &  \BFUpsilon_k  \geq \BFzero &&\forall k \in \K\\
  \quad & \eqref{sub_lambda:proof}-\eqref{budget:proof},
\end{align*}
%\end{subequations}
where each $\BFUpsilon_k \in \Re^{\lvert E \rvert}$ is a linearization of $\diag(\BFell)\BFlambda_k$. 
%\end{proof}
%   %%%%%----
%   %%%%----
%   %%%%%---

%\subsection{Proof of Theorem \ref{th:converge_sBB}}
% % to be written...
% % 
  \section{Matrices}\label{sec:appen:matrix}
  The coefficient matrices in \eqref{Compact_RMP} are given by: 
  \begin{footnotesize}
\begin{align*}
&\BFh= \begin{blockarray}{(cc)}
\BFzero\\\BFzero\\\BFzero\\\BFzero\\\BFzero\\ 1
\end{blockarray},\;\quad  W_{\BFell} = \begin{blockarray}{(ccc)c} -\BFI & \BF{f}_{\BFell}& -\BFone & \eqref{eq:DRO_Delta:n}\\ 
-\BFI & \BFzero & \BFzero &\eqref{eq:DRO_Delta_positive:n}\\
0&-1&0&\eqref{eq:DRO_u_positive:n}\\0 & \zeta_{\mymbox{lb}} & -1 &\eqref{eq:mccor1:n}\\  0 & -\zeta_{\mymbox{ub}} & 1&\eqref{eq:mccor2:n}\end{blockarray},\\   
  & A= \begin{blockarray}{(cccccc)c}
\BFone\hat{\BFq}^T-\BFI& \BFone \BFone^T & \BFone \BFone^T & \BFone & \Gamma\BFone & -\BFone&\eqref{eq:DRIP:RC:C1:n}\\
    \diag(\bar{\BFq})& -\BFI & 0 &0& -\BFone & 0&\eqref{eq:chi_plus:n}\\
    -\diag(\bar{\BFq})& 0&-\BFI& 0 &-\BFone & 0&\eqref{eq:chi_minus:n}\\
     0 & 0 & 0 & 0 &0 &0&\eqref{eq:DRO_u_sum_one:relaxed}\\
     0 & 0 & 0 & 0 &0 & 0&\eqref{eq:DRO_u_sum_one:relaxed}\\
     0 & 0 & 0 & -1 &0 & 0&\eqref{RRLT:relaxed} \\
  0 & 0 & 0 & 1 &0 & 0&\eqref{RRLT:relaxed} \\
      0 & 0 & 0 & \BFone  & 0 &0 &\eqref{eq:mccor3:n}\\
 0 & 0 & 0 & -\BFone &0  & 0&\eqref{eq:mccor4:n}\\
      0 & 0 & 0 & 1&0 & 0 &\eqref{eq:DRO_zeta_bound:n}\\
 0 & 0 & 0 & -1 &0 & 0&\eqref{eq:DRO_zeta_bound:n}\\
 0 & -\BFI & 0 & 0 &0 & 0 &\eqref{eq:w:chi:n}\\
  0 & 0 & -\BFI & 0 &0 & 0 &\eqref{eq:w:chi:n}\\
 0 & 0 & 0 & 0 &-1 & 0&\eqref{eq:w:chi:n} \\
    \end{blockarray},\;
    %%%%%
    %%%%----
    %%----
\;\; B_{\BFell}=\begin{blockarray}{(ccc)}
    \frac{1}{1-\alpha}\BFI& 0 & 0\\
     0 &  0 & 0\\
     0 &  0 & 0\\
     0 & 1 & 0\\
      0 & -1 & 0\\
  0 & 0& 1\\
      0 & 0 & -1\\
  0& \zeta_{\mymbox{ub}}\BFe_{\BFell} & -\BFe_{\BFell}\\
      0&-\zeta_{\mymbox{lb}}\BFe_{\BFell} & \BFe_{\BFell} \\
      0 & 0 & 0\\
      0&0 &0\\
      0 & 0 & 0\\
      0 & 0 & 0\\
        0 & 0 & 0\\
    \end{blockarray},\;\; \BFs= \begin{blockarray}{(ccc)}0\\
    0\\ 0 \\  1\\ -1\\ 0\\ 0\\\zeta_{\mymbox{ub}} \BFone\\ -\zeta_{\mymbox{lb}} \BFone\\ \zeta_{\mymbox{ub}}\\ -\zeta_{\mymbox{lb}}\\ 0\\ 0\\
    0\end{blockarray}.
\end{align*}
\end{footnotesize}
where for each $\BFell \in \LL$,  $\BF{f}_{\BFell} \in \Re^{\NK}$ is the vector of maximum flows for each scenario in $\NK$.

\section{ Solving for deterministic strategy}\label{appen:deterministic}
Since $\zeta$ is bounded above by $\bar{\zeta}$,  we can linearize the relation $\eta_\BFell=u_\BFell \zeta$ in \eqref{eq:DRIP:RC} to obtain the following MILP:
%\begin{subequations}
  \begin{align*}
\minimize_{\scriptsize \begin{array}{c}\BFu, t,\BFeta,\BFw,\BFw^-\\\chi, \BFbeta, \BFDelta, \zeta\end{array}} \quad & t\\ 
\subto \quad &  \zeta+\sum_{k'\in \K}w_{k'}+\sum_{k'\in \K}w_{k'}^-+\Gamma\chi \notag \\\qquad&+\sum_{k'\in\K} 
 \hat{q}_{k'}\beta_{k'}+\frac{1}{1-\alpha}\sum\limits_{\BFell \in \LL} \Delta_{\BFell,  k}-\beta_k \leq t && \forall k\in\K \,\\ 
 \quad & \chi\geq  \bar{q}_k\beta_k-w_k&& \forall k \in \K \\
    \quad & \chi\geq -\bar{q}_k\beta_k-w^-_k&&\forall k 
    \in \K \\
 \quad & \BFw\geq 0,\BFw^-\geq 0,\chi \geq 0\\
  \quad &   \Delta_{\BFell,  k} \geq u_\BFell f_{\BFell, k}-\eta_\BFell && \forall \BFell  \in \LL, k\in \K   \\
  \quad &   \Delta_{\BFell,  k} \geq 0 && \forall \BFell  \in \LL, k\in \K   \\
  \quad &   \eta_{\BFell} \geq 0 && \forall \BFell  \in \LL  \\
  \quad & \eta_{\BFell} \geq  \zeta  -(1-u_\BFell) \bar{\zeta}&& \forall \BFell  \in \LL \\
  \quad &   \eta_{\BFell}\leq \zeta  && \forall \BFell  \in \LL\\
  \quad &   \eta_{\BFell}\leq \bar{\zeta} u_\BFell  && \forall \BFell  \in \LL\\
\quad &  \BFu = \{0,1\}^{\lvert \LL \rvert }\\
\quad & \BFone^T \BFu=1.
  \end{align*}

\end{document}